\documentclass[11pt,twoside,a4paper,notitlepage]{article}

\usepackage{ifthen}

\newboolean{conferenceversion}
\setboolean{conferenceversion}{false}

\ifthenelse{\boolean{conferenceversion}}
{}
{
  \usepackage[a4paper,margin=2.5cm]{geometry}
  \usepackage{times}
}

\ifthenelse{\boolean{conferenceversion}}
{}
{
  \usepackage[english]{babel}
}

\ifthenelse{\boolean{conferenceversion}}
{}
{
  \usepackage{fancyhdr}
}

\usepackage[sf,SF]{subfigure}

\usepackage{bussproofs}

\ifthenelse{\boolean{conferenceversion}}
{}
{
  
}

\usepackage{ifthen}

\usepackage{ifthen}

\newboolean{detectedSTOC}
\newboolean{detectedFOCS}
\newboolean{detectedElsevier}
\newboolean{detectedNOW}
\newboolean{detectedLMCS}
\newboolean{detectedPoster}
\newboolean{detectedSIAM}
\newboolean{detectedIEEE}
\newboolean{detectedLNCS}
\newboolean{detectedACM}
\newboolean{detectedToC}
\newboolean{detectedLIPIcs}
\newboolean{detectedArticle}
\newboolean{detectedReport}
\newboolean{detectedThesis}

\makeatletter

\@ifclassloaded{sig-alternate}
{\setboolean{detectedSTOC}{true}}
{\setboolean{detectedSTOC}{false}}

\@ifclassloaded{elsarticle}
{\setboolean{detectedElsevier}{true}}
{\setboolean{detectedElsevier}{false}}

\@ifclassloaded{now}
{\setboolean{detectedNOW}{true}}
{\setboolean{detectedNOW}{false}}

\@ifclassloaded{lmcs}
{\setboolean{detectedLMCS}{true}}
{\setboolean{detectedLMCS}{false}}

\@ifclassloaded{IEEEtran} {
  \setboolean{detectedIEEE}{true}
  \ifCLASSOPTIONconference {
    \setboolean{detectedFOCS}{true}
  }
  \else {
    \setboolean{detectedFOCS}{false}
  }
  \fi
}
{
  \setboolean{detectedFOCS}{false}
  \setboolean{detectedIEEE}{false}
}

\@ifclassloaded{siamltex1213} 
{\setboolean{detectedSIAM}{true}}
{\setboolean{detectedSIAM}{false}}

\@ifclassloaded{llncs}
{\setboolean{detectedLNCS}{true}}
{\setboolean{detectedLNCS}{false}}

\@ifclassloaded{acmsmall}
{\setboolean{detectedACM}{true}}
{\setboolean{detectedACM}{false}}

\@ifclassloaded{toc}
{\setboolean{detectedToC}{true}}
{\setboolean{detectedToC}{false}}

\@ifclassloaded{lipics}
{\setboolean{detectedLIPIcs}{true}}
{\setboolean{detectedLIPIcs}{false}}

\@ifclassloaded{sciposter}
{\setboolean{detectedPoster}{true}}
{\setboolean{detectedPoster}{false}}

\@ifclassloaded{article}
{\setboolean{detectedArticle}{true}}
{\setboolean{detectedArticle}{false}}

\makeatother

\ifthenelse{\not \isundefined{\examen} 
  \and \not \isundefined{\disputationsdatum} 
  \and \not \isundefined{\disputationslokal}}   
  {\setboolean{detectedThesis}{true}}
  {\setboolean{detectedThesis}{false}}
           
\ifthenelse{\boolean{detectedArticle} \or \boolean{detectedThesis}
  \or \boolean{detectedSTOC} \or \boolean{detectedFOCS}
  \or \boolean{detectedSIAM} \or \boolean{detectedIEEE}
  \or \boolean{detectedPoster}}
{\setboolean{detectedReport}{false}}
{\setboolean{detectedReport}{true}}

\ifthenelse{\boolean{detectedLIPIcs}}
{}
{\usepackage[utf8]{inputenc}}

\usepackage{amsmath}
\usepackage{amssymb}
\usepackage{amsfonts}

\usepackage{mathtools}

\ifthenelse
{\boolean{detectedElsevier}}
{}
{\usepackage{varioref}}

\usepackage{xspace}

\ifthenelse    
{\boolean{detectedSTOC}   \or \boolean{detectedSIAM} \or 
  \boolean{detectedLMCS}  \or \boolean{detectedNOW}  \or 
  \boolean{detectedLNCS}  \or \boolean{detectedACM}  \or
  \boolean{detectedToC}  \or \boolean{detectedLIPIcs}}
{}
{
  \usepackage{amsthm}
  \usepackage{amsthm-boldfont}
}

\ifthenelse        
{\boolean{detectedElsevier}  \or \boolean{detectedFOCS}   \or 
  \boolean{detectedSTOC}     \or \boolean{detectedPoster} \or
  \boolean{detectedSIAM}     \or \boolean{detectedLMCS}   \or
  \boolean{detectedIEEE}     \or \boolean{detectedNOW}    \or
  \boolean{detectedToC}      \or \boolean{detectedThesis} \or
  \boolean{detectedLNCS}     \or \boolean{detectedACM}    \or 
  \boolean{detectedLIPIcs}}
{}
{\usepackage{nada-typography}}

\usepackage{ifthen}
\usepackage{xspace}
\ifthenelse
{\boolean{detectedElsevier}}
{}
{\usepackage{varioref}}

\DeclareMathAlphabet{\mathsfsl}{OT1}{cmss}{m}{sl}

\newcommand{\formatfunctiontoset}[1]{\mathit{#1}}

\newcommand{\introduceterm}[1]{{\emph{#1}}}

\newcommand{\eqperiod}{\enspace .}
\newcommand{\eqcomma}{\enspace ,}

\newcommand{\wrt}{with respect to\xspace}

\ifthenelse{\boolean{detectedToC}}{}
  {\newcommand{\eg}{for instance\xspace}     }
     
\ifthenelse{\boolean{detectedToC}}{}{\newcommand{\ie}{i.e.,\ }}

\ifthenelse{\boolean{detectedLIPIcs}}
{}
{}

\newcommand{\etal}{et al.\@\xspace}

\ifthenelse{\boolean{detectedIEEE}}{}{}

\newcommand{\Bigoh}[1]{\mathrm{O} \bigl( #1 \bigr)}
\newcommand{\bigoh}[1]{\mathrm{O} ( #1 )}

\newcommand{\bigtheta}[1]{\Theta ( #1 )}

\newcommand{\Bigomega}[1]{\Omega \bigl( #1 \bigr)}
\newcommand{\bigomega}[1]{\Omega ( #1 )}

\newcommand{\complclassformat}[1]        {\textrm{\upshape{\textsf{#1}}}\xspace}

\newcommand{\NP}{\complclassformat{NP}}

\newcommand{\refsec}[1]{Section~\ref{#1}}

\newcommand{\refth}[1]{Theorem~\ref{#1}}

\newcommand{\refthm}[1]{Theorem~\ref{#1}}

\newcommand{\reflem}[1]{Lemma~\ref{#1}}

\newcommand{\refpr}[1]{Proposition~\ref{#1}}

\newcommand{\refdef}[1]{Definition~\ref{#1}}

\newcommand{\refrem}[1]{Remark~\ref{#1}}

\newcommand{\refobs}[1]{Observation~\ref{#1}}

\newcommand{\refex}[1]{Example~\ref{#1}}

\newcommand{\Refth}[1]{Theorem~\ref{#1}}

\ifthenelse
{\isundefined{\refeq}}
{\newcommand{\refeq}[1]{\eqref{#1}}}
{\renewcommand{\refeq}[1]{\eqref{#1}}}

\ifthenelse{\boolean{detectedToC}}{}
{
  
  \newcommand{\R}         {\mathbb{R}}
  
  \newcommand{\N}         {\mathbb{N}}
  \newcommand{\Nplus}     {\mathbb{N}^{+}}

}

\newcommand{\F}{\mathbb{F}}

\ifthenelse{\boolean{detectedLMCS}}
{}
{}

\newcommand{\twincommandJN}[6]    {#1#2#3\vphantom{#2#5}\mspace{-2.25mu}#4.#5#6}

\newcommand{\funcdescr}[3]{\ensuremath{ #1 : #2 \to #3}}

\newcommand{\edges}[1]{E( #1 )}

\newcommand{\descname}{\formatfunctiontoset{desc}}

\ifthenelse{\isundefined{\descnode}}
{\newcommand{\descnode}[2][]{\descname_{#1}(#2)}}
{\renewcommand{\descnode}[2][]{\descname_{#1}(#2)}}

\newcommand{\set}[1]{\{ #1 \}}

\newcommand{\Setdescr}[3][|]     {\twincommandJN{\bigl\{}{#2}{\bigl#1}{\bigr}{\,#3}{\bigr\}}}

\newcommand{\setsize}[1]{\lvert#1\rvert}

\newcommand{\union}{\cup}

\newcommand{\disjointunion}{\overset{.}{\cup}}

\newcommand{\Lor}{\bigvee}
\newcommand{\Land}{\bigwedge}

\newcommand{\olnot}[1]{\overline{#1}}

\newcommand{\inductionformat}[1]{\textit{#1}}

\newcommand{\BASE}[1][]
        {\inductionformat
                {                        \ifthenelse{\equal{#1}{}}                                {Base case: }                                {Base case (#1):}                }        }

\ifthenelse{\isundefined{\DONOTLOADHYPERREFPACKAGE}
  \and \not \boolean{detectedSTOC}     \and \not \boolean{detectedFOCS}
  \and \not \boolean{detectedElsevier} \and \not \boolean{detectedPoster} 
  \and \not \boolean{detectedSIAM}     \and \not \boolean{detectedACM}
  \and \not \boolean{detectedIEEE}     \and \not \boolean{detectedNOW}
  \and \not \boolean{detectedToC}      \and \not \boolean{detectedThesis}
  \and \not \boolean{detectedLNCS}     \and \not \boolean{detectedLIPIcs}}
{\usepackage[colorlinks=true,citecolor=blue]{hyperref}}
{}

\ifthenelse{\boolean{detectedSTOC}}                  
  {

\newtheorem{theorem}{Theorem}
\newtheorem{lemma}[theorem]{Lemma}
\newtheorem{proposition}[theorem]{Proposition}
\newtheorem{corollary}[theorem]{Corollary}
\newtheorem{observation}[theorem]{Observation}

\newtheorem{definition}[theorem]{Definition}
\newtheorem{claim}[theorem]{Claim}

\newtheorem{conjecture}{Conjecture}
\newtheorem{openproblem}[conjecture]{Open Problem}

\newdef{remark}{Remark}
\newdef{example}{Example}

\newcounter{unnumber}

 }
  {\ifthenelse{\boolean{detectedSIAM}}
    {

\newtheorem{observation}[theorem]{Observation}

\newtheorem{conjecture}{Conjecture}
\newtheorem{openquestion}{Open Question}

\newtheorem{remarkinner}[theorem]{Remark}
\newtheorem{exampleinner}[theorem]{Example}

\newcommand{\exampleendmarker}{\qquad$\Diamond$}
\newcommand{\remarkendmarker}{\qquad$\Diamond$}

\newenvironment{example}                        
    {\begin{exampleinner} \rm}
    {\exampleendmarker\end{exampleinner}}

\newenvironment{remark}                        
    {\begin{remarkinner} \rm}
    {\remarkendmarker\end{remarkinner}}

\newcounter{unnumber}

 }
    {\ifthenelse{\boolean{detectedNOW}}
      {

\newtheorem{standardlocalcounter}{Dummy}[chapter]
\newtheorem{standardglobalcounter}{Dummy}

\newtheorem{theorem}[standardlocalcounter]{Theorem}
\newtheorem{lemma}[standardlocalcounter]{Lemma}
\newtheorem{proposition}[standardlocalcounter]{Proposition}
\newtheorem{corollary}[standardlocalcounter]{Corollary}
\newtheorem{observation}[standardlocalcounter]{Observation}
\newtheorem{fact}[standardlocalcounter]{Fact}

\newtheorem{conjecturelocalcounter}[standardlocalcounter]{Conjecture}
\newtheorem{conjectureglobalcounter}[standardglobalcounter]{Conjecture}
\newtheorem{conjecture}[standardglobalcounter]{Conjecture}
\newtheorem{openquestion}[standardglobalcounter]{Open Question}
\newtheorem{openproblem}[standardglobalcounter]{Open Problem}
\newtheorem{problem}{Problem}

\newtheorem{property}[standardlocalcounter]{Property}
\newtheorem{definition}[standardlocalcounter]{Definition}
\newtheorem{claim}[standardlocalcounter]{Claim}

\newtheorem{algorithm}[standardlocalcounter]{Algorithm}

\newtheorem{remark}[standardlocalcounter]{Remark}
\newtheorem{example}[standardlocalcounter]{Example}

\renewenvironment{proof}[1][Proof]{\par\trivlist
   \item[\hskip \labelsep{\itshape {#1}.}]\prooffont}
   {\hspace*{0pt plus1fill}\fboxsep2.5pt\fboxrule.5pt\raise3pt\hbox{\fbox{}}\endtrivlist}
 }
      {\ifthenelse{\boolean{detectedLMCS}}
        {

\theoremstyle{plain}    

\newtheorem{theorem}[thm]{Theorem}
\newtheorem{lemma}[thm]{Lemma}
\newtheorem{proposition}[thm]{Proposition}
\newtheorem{corollary}[thm]{Corollary}
\newtheorem{observation}[thm]{Observation}

\newtheorem{conjecture}[thm]{Conjecture}
\newtheorem{problem}[thm]{Problem}

\newtheorem{openquestion}{Open Question}
\newtheorem{openproblem}{Open Problem}

\theoremstyle{definition}

\newtheorem{property}[thm]{Property}
\newtheorem{definition}[thm]{Definition}
\newtheorem{claim}[thm]{Claim}
\newtheorem{remark}[thm]{Remark}
\newtheorem{example}[thm]{Example}

 }
        {\ifthenelse{\boolean{detectedIEEE}}
          {

\newtheorem{standardlocalcounter}{Dummy}[section]
\newtheorem{standardglobalcounter}{Dummy}

\theoremstyle{plain}    

\newtheorem{theorem}[standardglobalcounter]{Theorem}
\newtheorem{lemma}[standardglobalcounter]{Lemma}
\newtheorem{proposition}[standardglobalcounter]{Proposition}
\newtheorem{corollary}[standardglobalcounter]{Corollary}
\newtheorem{observation}[standardglobalcounter]{Observation}
\newtheorem{fact}[standardglobalcounter]{Fact}

\newtheorem{conjecture}[standardglobalcounter]{Conjecture}
\newtheorem{openquestion}{Open Question}
\newtheorem{openproblem}{Open Problem}
\newtheorem{problem}{Problem}

\theoremstyle{definition}

\newtheorem{property}[standardglobalcounter]{Property}
\newtheorem{definition}[standardglobalcounter]{Definition}
\newtheorem{claim}[standardglobalcounter]{Claim}

\theoremstyle{remark}
\newtheorem{remark}[standardglobalcounter]{Remark}
\newtheorem{example}[standardglobalcounter]{Example}

\newtheoremstyle{meta}  {3pt}  {3pt}  {\scshape \small }  {}  {\scshape \small }  {:}  { }          {}
\theoremstyle{meta}
\newtheorem{meta}{Meta comment}

\newtheoremstyle{questions}  {3pt}  {3pt}  {\sffamily \slshape}  {}  {\bfseries \sffamily \slshape}  {:}  { }          {}
\theoremstyle{questions}
\newtheorem{questions}{Open questions}

 }
          {\ifthenelse{\boolean{detectedLNCS}}
            {

\spnewtheorem*{proofsketch}{Proof sketch}{\itshape}{\rmfamily}

\spnewtheorem{observation}{Observation}{\bfseries}{\itshape}
\spnewtheorem{fact}{Fact}{\bfseries}{\itshape}

 }
            {\ifthenelse{\boolean{detectedACM}}
              {

\newtheorem{conjecture}[theorem]{Conjecture}
\newtheorem{observation}[theorem]{Observation}
\newtheorem{claim}[theorem]{Claim}
\newtheorem{openquestion}{Open Question}

\newcounter{unnumber}

 }
              {\ifthenelse{\boolean{detectedToC}}
                {

\theoremstyle{plain}
\newtheorem{observation}[theorem]{Observation}
\newtheorem{openproblem}[theorem]{Open Problem}

\theoremstyle{definition}
\newtheorem{property}[theorem]{Property}

\renewcommand{\refth}[1]{\expref{Theorem}{#1}}
\renewcommand{\reflem}[1]{\expref{Lemma}{#1}}
\renewcommand{\refpr}[1]{\expref{Proposition}{#1}}

\renewcommand{\refdef}[1]{\expref{Definition}{#1}}
\renewcommand{\refrem}[1]{\expref{Remark}{#1}}
\renewcommand{\refobs}[1]{\expref{Observation}{#1}}

\renewcommand{\refex}[1]{\expref{Example}{#1}}

\renewcommand{\Refth}[1]{\expref{Theorem}{#1}}

\renewcommand{\refsec}[1]{\expref{Section}{#1}}

 }
                {\ifthenelse{\boolean{detectedLIPIcs}}
                  {

\theoremstyle{plain}    

\newtheorem{fact}[theorem]{Fact}
\newtheorem{proposition}[theorem]{Proposition}
\newtheorem{observation}[theorem]{Observation}

 }
                  {\ifthenelse{\boolean{detectedArticle} 
                      \or \boolean{detectedElsevier}}
                    {

\newtheorem{standardlocalcounter}{Dummy}[section]
\newtheorem{standardglobalcounter}{Dummy}

\theoremstyle{plain}    

\newtheorem{theorem}[standardlocalcounter]{Theorem}
\newtheorem{lemma}[standardlocalcounter]{Lemma}
\newtheorem{proposition}[standardlocalcounter]{Proposition}
\newtheorem{corollary}[standardlocalcounter]{Corollary}
\newtheorem{observation}[standardlocalcounter]{Observation}
\newtheorem{fact}[standardlocalcounter]{Fact}

\newtheorem{conjecturelocalcounter}[standardlocalcounter]{Conjecture}
\newtheorem{conjectureglobalcounter}[standardglobalcounter]{Conjecture}
\newtheorem{conjecture}[standardglobalcounter]{Conjecture}
\newtheorem{openquestion}[standardglobalcounter]{Open Question}
\newtheorem{openproblem}[standardglobalcounter]{Open Problem}
\newtheorem{problem}[standardglobalcounter]{Problem}

\theoremstyle{definition}

\newtheorem{property}[standardlocalcounter]{Property}
\newtheorem{definition}[standardlocalcounter]{Definition}
\newtheorem{claim}[standardlocalcounter]{Claim}

\theoremstyle{remark}
\newtheorem{remark}[standardlocalcounter]{Remark}
\newtheorem{example}[standardlocalcounter]{Example}

\newtheoremstyle{meta}  {3pt}  {3pt}  {\scshape \small }  {}  {\scshape \small }  {:}  { }          {}
\theoremstyle{meta}
\newtheorem{meta}{Meta comment}

\newtheoremstyle{questions}  {3pt}  {3pt}  {\sffamily \slshape}  {}  {\bfseries \sffamily \slshape}  {:}  { }          {}
\theoremstyle{questions}
\newtheorem{questions}{Open questions}

 }
                    {

\newtheorem{standardlocalcounter}{Dummy}[chapter]
\newtheorem{standardglobalcounter}{Dummy}

\theoremstyle{plain}    

\newtheorem{theorem}[standardlocalcounter]{Theorem}
\newtheorem{lemma}[standardlocalcounter]{Lemma}
\newtheorem{proposition}[standardlocalcounter]{Proposition}

\newtheorem{observation}[standardlocalcounter]{Observation}

\theoremstyle{definition}

\newtheorem{definition}[standardlocalcounter]{Definition}

\theoremstyle{remark}
\newtheorem{remark}[standardlocalcounter]{Remark}
\newtheorem{example}[standardlocalcounter]{Example}

\newtheoremstyle{meta}  {3pt}  {3pt}  {\scshape \small }  {}  {\scshape \small }  {:}  { }          {}
\theoremstyle{meta}

\newtheoremstyle{questions}  {3pt}  {3pt}  {\sffamily \slshape}  {}  {\bfseries \sffamily \slshape}  {:}  { }          {}
\theoremstyle{questions}

 }}}}}}}}}}

\ifthenelse{\boolean{detectedThesis}}
{}
{\usepackage{ifpdf}}

\ifthenelse{\boolean{detectedToC}}
{}
{\usepackage{graphicx}}

\ifpdf         
\fi

\ifthenelse
{\boolean{detectedSTOC}  \or \boolean{detectedThesis} \or 
 \boolean{detectedLNCS}  \or \boolean{detectedToC}    \or 
 \boolean{detectedLIPIcs}}    
{}
{
  \setcounter{secnumdepth}{3}
  \setcounter{tocdepth}{3}
}

\newcount\hours
\newcount\minutes
\def\SetTime{\hours=\time
\global\divide\hours by 60
\minutes=\hours
\multiply\minutes by 60
\advance\minutes by-\time
\global\multiply\minutes by-1 }
\SetTime
\def\now{\number\hours:\ifnum\minutes<10 0\fi\number\minutes}

\DeclareFontFamily{OT1}{pzc}{}
\DeclareFontShape{OT1}{pzc}{m}{it}{<-> s * [1.200] pzcmi7t}{}
\DeclareMathAlphabet{\mathpzc}{OT1}{pzc}{m}{it}

\newcommand{\fstd}{{\ensuremath{F}}}

\newcommand{\lita}{\ensuremath{a}}

\newcommand{\cla}{\ensuremath{A}}
\newcommand{\clb}{\ensuremath{B}}
\newcommand{\clc}{\ensuremath{C}}
\newcommand{\cld}{\ensuremath{D}}

\newcommand{\setsofvarsorlitsmall}[2]
        {\mathit{#1}({#2})}

\newcommand{\vars}[1]{\setsofvarsorlitsmall{Vars}{#1}}

\newcommand{\restrict}[2]{#1\!\!\upharpoonright_{#2}}

\newcommand{\sizestd}{{\size}}

\newcommand{\stoptime}{\tau}

\usepackage{stmaryrd}

\newcommand{\varset}{\vars{\fstd}}

\newcommand{\myvec}[1]           {\ifthenelse{\equal{#1}{j}}
             {\vec{\jmath}\,}
             {\ifthenelse{\equal{#1}{i}}
               {\vec{\imath}\,}
               {\vec{#1}}}}

\newcommand{\sumsofsquares}{sums-of-squares\xspace}
\newcommand{\Sumsofsquares}{Sums-of-squares\xspace}
\newcommand{\SOS}{SOS\xspace}

\newcommand{\polyequation}{f}
\newcommand{\polymultiplier}{g}
\newcommand{\polyinequality}{h}
\newcommand{\polysos}{u}
\newcommand{\polysoscomp}{q}

\newcommand{\polyequationcount}{s}
\newcommand{\polyinequalitycount}{t}

\newcommand{\polytheorem}{p}

\newcommand{\sosdegree}{degree\xspace}
\newcommand{\sossize}{size\xspace}
\newcommand{\reswidth}{width\xspace}
\newcommand{\ressize}{size\xspace}
\renewcommand{\sizestd}{S}
\newcommand{\degreestd}{d}

\newcommand{\ressizestd}{\sizestd}
\newcommand{\reswidthstd}{w}
\newcommand{\sossizestd}{\sizestd}
\newcommand{\sosdegreestd}{\degreestd}

\newcommand{\indexdegree}{domain-degree\xspace}

\newcommand{\indexwidth}{domain-width\xspace}

\newcommand{\indexdegreestd}{\degreestd}
\newcommand{\indexdegreerun}{\eta}
\renewcommand{\indexdegreerun}{s}

\newcommand{\monomialstd}{M}

\newcommand{\indicatorbase}{\delta}
\newcommand{\indicator}[2]{\indicatorbase_{#1=#2}}
\newcommand{\stringstd}{\beta}
\newcommand{\indicatorstd}{\indicator{\myvec{\varstd}}{\stringstd}}

\newcommand{\sumencoding}[1]{S(#1)}

\newcommand{\indexstd}{i}
\newcommand{\indexaux}{j}
\newcommand{\indexauxaux}{\ell}
\newcommand{\genericint}{\ell}

\newcommand{\formulastd}{F}
\newcommand{\varstd}{x}
\newcommand{\varnum}{n}
\newcommand{\litstd}{\lita}

\newcommand{\clausecount}{\polyinequalitycount}

\newcommand{\clauseidx}{\indexaux}
\newcommand{\clauseidxaux}{\indexauxaux}
\newcommand{\varidx}{\indexstd}
\newcommand{\litidx}{\indexstd}

\newcommand{\graphstd}{G}
\newcommand{\cliquesize}{k}
\newcommand{\cliqueformula}[2][\cliquesize]    {\ensuremath{{#1}\text{-}\mathrm{Clique}(#2)}}
\newcommand{\blockformula}[1]    {\ensuremath{\cliquesize\text{-}\mathrm{Block}(#1)}}

\newcommand{\Cliqueformula}[2][\cliquesize]    {\ensuremath{{#1}\text{-}\mathrm{Clique}\bigl(#2\bigr)}}
\newcommand{\Blockformula}[1]    {\ensuremath{\cliquesize\text{-}\mathrm{Block}\bigl(#1\bigr)}}

\newcommand{\cliquevar}{x}
\newcommand{\cliqueaux}{z}
\newcommand{\cliquevertexnum}{N}
\newcommand{\cliquevertexbound}{N_0}
\newcommand{\cliquedegreebound}{\xordegreelb_0}

\newcommand{\xorstd}{\phi}
\newcommand{\xorformula}{$3$-XOR\xspace}
\newcommand{\xorvar}{x}
\newcommand{\xorvarA}{x}
\newcommand{\xorvarB}{y}
\newcommand{\xorvarC}{z}
\newcommand{\xorvarnum}{n}

\newcommand{\xordegreelb}{\alpha}
\newcommand{\xorminvar}{\xorvarnum_0}
\renewcommand{\xorminvar}{\xorvarnum_\epsilon}

\newcommand{\graphclauseparam}[2]{          G^{#1}_{#2}
}
\newcommand{\graphclausestd}{\graphclauseparam{\cliquesize}{\xorstd}}

\newcommand{\blockvarnum}{t}

\newcommand{\domainsize}{t}
\newcommand{\domainset}{D}
\newcommand{\formulaparam}[1]{\formulastd[{#1}]}

\newcommand{\domainidx}{\indexstd}

\newcommand{\signature}[1]{(\domainidx_1,\domainidx_2,\ldots,\domainidx_{#1})}
\renewcommand{\signature}[1]    {\set{\domainidx_1,\domainidx_2,\ldots,\domainidx_{#1}}}
\newcommand{\signaturename}{\myvec{\domainidx}}

\newcommand{\domainlarge}{m}
\newcommand{\domainsmall}{k}

\newcommand{\relativization}{relativization\xspace}
\newcommand{\Relativization}{Relativization\xspace}

\newcommand{\targetformula}{\formulastd_{\domainlarge,\domainsmall}}

\newcommand{\relformula}[2]{\formulaparam{#1;#2}}
\newcommand{\relstd}{\relformula{\domainsmall}{\domainlarge}}

\newcommand{\relguard}{s}
\newcommand{\guardvariable}{selector\xspace}
\newcommand{\guardedclause}{selectable clause\xspace}

\newcommand{\lowerboundlimit}{\delta}

\newcommand{\THR}[2]{\mathsf{Thr}_{#1}(#2)}

\newcommand{\thrvarstd}{p}
\newcommand{\thrvaraux}{y}
\newcommand{\thridxaux}{\indexauxaux}

\newcommand{\partassign}{\rho}
\newcommand{\padistribution}{\mathcal{R}}

\newcommand{\RestrictionBoundDisplay}[1]    {\frac{{(4\domainsmall\log \domainlarge)}^{\domainsmall}}{\domainlarge^{#1}}}

\newcommand{\RestrictionBoundInline}[1]    {{(4\domainsmall\log \domainlarge)}^{\domainsmall}/{\domainlarge^{#1}}}

\newcommand{\InverseRestrictionBoundDisplay}[1]    {\frac{\domainlarge^{#1}}{{(4\domainsmall\log\domainlarge)}^{\domainsmall}}}

\newcommand{\InverseRestrictionBoundInline}[1]    {\domainlarge^{#1}/{(4\domainsmall\log \domainlarge)}^{\domainsmall}}

\newcommand{\TheauthorJN}{The second author\xspace}

\ifthenelse{\boolean{conferenceversion}}
{

}
{

}

\ifthenelse{\boolean{conferenceversion}}
{}
{
\newtheoremstyle{metacommenttheoremstyle}    {3pt}    {3pt}    {\sffamily \itshape \scriptsize
          }    {}    {\bfseries \scshape \footnotesize }    {:}    { }        {}
\theoremstyle{metacommenttheoremstyle}
}
\newtheorem{jncommentcontainer}{Jakob's comment}
\newtheorem{mlcommentcontainer}{Massimo's comment}

\ifthenelse{\isundefined{\DONOTINSERTCOMMENTS}}
{  \usepackage[usenames,dvipsnames]{xcolor}

  \newcommand{\jncomment}[1]  {\begin{jncommentcontainer} \textcolor{blue}{#1} \end{jncommentcontainer}}

  \newcommand{\mlcomment}[1]  {\begin{mlcommentcontainer} \textcolor{OliveGreen}{#1} \end{mlcommentcontainer}}

}
{  \newcommand{\jncomment}[1]{}
  \newcommand{\mlcomment}[1]{}
}

\ifthenelse{\boolean{conferenceversion}}
{}
{
  \numberwithin{equation}{section}
}

\begin{document}

\title{Tight Size-Degree Bounds for Sums-of-Squares Proofs  \thanks{This is the full-length version of the paper with the same
    title to appear in 
    \emph{Proceedings of the 30th Annual Computational Complexity
      Conference ({CCC}~'15)}.}}

\author{  Massimo Lauria \\
  KTH Royal Institute of Technology
  \and
  Jakob Nordström \\
  KTH Royal Institute of Technology}

\date{\today}

\maketitle

\ifthenelse{\boolean{conferenceversion}}
{}
{

    \thispagestyle{empty}

    \pagestyle{fancy}
    \fancyhead{}
  \fancyfoot{}
    \renewcommand{\headrulewidth}{0pt}
  \renewcommand{\footrulewidth}{0pt}
  
                    \fancyhead[CE]{\slshape 
    TIGHT SIZE-DEGREE BOUNDS FOR SUMS-OF-SQUARES PROOFS}
  \fancyhead[CO]{\slshape \nouppercase{\leftmark}}
  \fancyfoot[C]{\thepage}
  
            \setlength{\headheight}{13.6pt}
}

\begin{abstract}
  We exhibit families of $4$-CNF formulas over $n$~variables that have
  sums-of-squares (SOS) proofs of unsatisfia\-bility of degree
  (a.k.a.\ rank)~$d$ but require SOS proofs of size $n^{\bigomega{d}}$
  for values of $d = d(n)$ from constant all the way up
  to~$n^{\delta}$ for some universal constant~$\delta$. This shows
  that the $n^{\bigoh{d}}$ running time obtained by using the Lasserre
  semidefinite programming relaxations to find degree-$d$ SOS proofs
  is optimal up to constant factors in the exponent.  We establish
  this result by combining \mbox{\NP-reductions} expressible as
  low-degree SOS derivations with the idea of relativizing CNF
  formulas in \mbox{[Kraj{\'\i}{\v{c}}ek~'04]} and \mbox{[Dantchev and
      Riis~'03]}, and then applying a restriction argument as in
  [\mbox{Atserias}, Müller, and Oliva~'13] and [Atserias, Lauria, and
    Nordström~'14].  This yields a generic method of amplifying SOS
  degree lower bounds to size lower bounds, and also generalizes the
  approach in [ALN14] to obtain size lower bounds for the proof
  systems resolution, polynomial calculus, and Sherali-Adams from
  lower bounds on width, degree, and rank, respectively.
\end{abstract}

\section{Introduction}
\label{sec:intro}

Let
$\polyequation_1, \ldots, \polyequation_\polyequationcount \in \R[\varstd_1, \dots, \varstd_\varnum]$
be real, multivariate polynomials. Then the Positivstellen\-satz
proven in~\cite{Krivine64Anneaux,Stengle73NullstellensatzPositivstellensatz}
says (as a special case) that the the system of equations
\begin{equation}
  \label{eq:schematic-eqsyst}
  \polyequation_1 = 0, \ldots, \polyequation_\polyequationcount=0  
\end{equation}
has no solution over
$\R^\varnum$
if and only if there exist polynomials
$\polymultiplier_{\clauseidx},
\polysoscomp_\clauseidxaux \in \R[\varstd_1, \dots, \varstd_\varnum]$
such that
\begin{equation}
  \label{eq:schematic-sos-proof}
  \sum_{\clauseidx=1}^{\polyequationcount} \polymultiplier_\clauseidx
  \polyequation_\clauseidx = 
  - 1 - \sum_{\clauseidxaux}\polysoscomp^{2}_{\clauseidxaux}
\eqperiod  
\end{equation}
That there can exist no solution given an expression of the
form~\refeq{eq:schematic-sos-proof} is clear, but what is more
interesting is that there always exists such an expression to certify
unsatisfiability. 
We refer to 
\refeq{eq:schematic-sos-proof}
as  a 
\mbox{\introduceterm{Positivstellensatz proof}{}}
or
\introduceterm{\Sumsofsquares (\SOS) proof}
of unsatisfiability,
or as an
\introduceterm{\SOS refu\-tation},\footnote{All proofs for systems of polynomial equations or for 
  formulas in conjunctive normal form (CNF)
  in this paper will be proofs of 
  unsatisfiability, and we will therefore use the two terms ``proof''
  and ``refutation'' interchangeably.}
of~\refeq{eq:schematic-eqsyst}.
We remark that
the Positivstellensatz also applies if we add inequalities
$\polyinequality_1 \geq 0, \ldots,
\polyinequality_{\polyinequalitycount} \geq 0$
to the system of equations
and allow terms~$
-\polyinequality_\clauseidx 
\sum_{\clauseidxaux}
\polysoscomp_{\clauseidx,\clauseidxaux}^{2}$
on the right-hand side in~\refeq{eq:schematic-sos-proof}.

The
\introduceterm{\sosdegree}\footnote{This is sometimes also referred to as the ``rank,'' but we
  will stick to the term ``degree'' in this paper.} 
 of an \SOS refutation is the maximal \sosdegree of
any $\polymultiplier_\clauseidx \polyequation_\clauseidx$. The search
for proofs of 
constant \sosdegree~$\sosdegreestd$ is \introduceterm{automatizable}
as shown in a sequence of works by 
Shor~\cite{Shor87Approach},
Nesterov~\cite{Nesterov00Squared},
Lasserre~\cite{Lasserre01Explicit},
and 
Parrilo~\cite{Parrilo00Thesis}.
What this means is that if there exists a \sosdegree-$d$ \SOS refutation for a
system of polynomial equalities (and inequalities) over~$\varnum$ variables,
then such a refutation can be found in polynomial time
${\varnum}^{\bigoh{\sosdegreestd}}$. Briefly, one can view
\refeq{eq:schematic-sos-proof}
as linear system of equations in the coefficients 
of~$\polymultiplier_\clauseidx$
and~$\polysos = \sum_\clauseidxaux \polysoscomp_\clauseidxaux^2$
with the added constraint \mbox{that $\polysos$ is} a sum of squares, and
such a system can be solved by semidefinite programming in
$\sosdegreestd/2$~rounds of the Lasserre SDP hierarchy.

In the last few years there has been renewed interest in
\sumsofsquares in the context of constraint satisfaction problems
(CSPs) and hardness of approximation, as witnessed by, for instance,
\cite{BBHKSZ12Hypercontractivity,OZ13Approximability,Tulsiani09CSPgaps}.
These works have highlighted the importance of \SOS
degree upper bounds for CSP approximability, and this is currently a
very active area of study. 

Our focus in this paper is not on algorithmic questions, however, but
more on \sumsofsquares viewed as a proof system (also referred to in the
literature as \introduceterm{Positivstellensatz} or
\introduceterm{Lasserre}). This proof system was 
introduced by Grigoriev and Vorobjov~\cite{GV01Complexity}
as an extension of the Nullstellensatz proof system studied by Beame
\etal~\cite{BIKPP94LowerBounds}, and Grigoriev established SOS   
degree lower bound for unsatisfiable $\F_2$-linear
equations~\cite{Grigoriev01LinearLowerBound} 
(also referred to as the \xorformula problem when each equation
involves at most $3$~variables)
and for the knapsack problem~\cite{Grigoriev01Knapsack}.

Given the connections to semidefinite programming and the
Lasserre SDP
hierarchy, it is perhaps not surprising that most works on SOS lower
bounds have focused on the degree measure.
However, from a proof complexity point of view it is also natural to
ask about the minimal \introduceterm{size} of SOS proofs, measured as
the number of monomials when all polynomials in each term
in~\refeq{eq:schematic-sos-proof} are expanded out as linear
combinations of monomials. Such SOS size lower bounds were proven
for knapsack in~\cite{GHP02ExponentialLowerBound}
and $\F_2$-linear systems of equations in~\cite{KI06LowerBounds},\footnote{It might be worth pointing out that definitions and terminology
  in this area have suffered from a certain lack of standardization,
  and so what \cite{KI06LowerBounds} refers to as
  ``static {L}ov{\'a}sz-{S}chrijver calculus''
  is closer to what we mean by SOS/Lasserre.}
and tree-like size lower bounds for other formulas were also obtained
in~\cite{PS12Exponential}.

A wider interest in this area of research was awakened when
Schoenebeck~\cite{Schoenebeck08LinearLevel} (essentially) rediscovered
Grigoriev's result~\cite{Grigoriev01LinearLowerBound}, which together
with further work by Tulsiani~\cite{Tulsiani09CSPgaps} led to
integrality gaps for a number of constraint satisfaction problems.
There have also been papers such as
\cite{BPS07LS}
and~\cite{GP14CommunicationLowerBounds}
focusing on \emph{semantic} versions of the proof system, with less
attention to the actual syntactic derivation rules used.
We refer the reader to, \eg, the introductory section
of~\cite{OZ13Approximability} for more background on \sumsofsquares
and connections to hardness of approximation, and to the survey
\cite{BS14SumsOfSquares} for an in-depth discussion of \SOS as an
approximation algorithm and the intriguing connections to the
so-called Unique Games Conjecture~\cite{Khot02PowerUnique}.

\subsection{Our Contribution}
\label{sec:intro-contributions}

As discussed above, if a system of polynomial equalities and
inqualities 
over~$\varnum$ variables
can be shown inconsistent by SOS in degree~$d$, then by using
semidefinite programming one can find an SOS refutation of the system
in time~${\varnum}^{\bigoh{\sosdegreestd}}$.
It is natural to ask whether this is optimal, or whether there might
exist ``shortcuts'' that could lead to SOS refutations more
quickly. 

We prove that there are no such shortcuts in general, but
that the
running time obtained by using the Lasserre
semidefinite programming relaxations to find SOS proofs
is optimal up to the constant in the exponent. 
We show this by constructing
formulas 
on $n$~variables 
(which can be translated to systems of polynomial equalities in a
canonical way) 
that have
\SOS refutations of degree~$\sosdegreestd$ 
but 
require refutations of size
$n^{\bigomega{\sosdegreestd}}$. Our lower bound proof works
for~$\sosdegreestd$ from 
constant all the way up to~$n^{\lowerboundlimit}$ for some
constant~$\lowerboundlimit$.

\begin{theorem}[informal]
  \label{th:main-result-informal}
  Let
  $
  \sosdegreestd =
  \sosdegreestd(\varnum)
  \leq \varnum^\lowerboundlimit
  $ 
  where
  $\lowerboundlimit > 0$
  is a universal constant.
  Then there is a family of $4$-CNF formulas
  $\set{\fstd_{\varnum}}_{\varnum \in \Nplus}$
  with
  $\Bigoh{\varnum^2}$~clauses over
  $\bigoh{\varnum}$~variables
  such that
  $\fstd_\varnum$
  is refutable in \sumsofsquares in degree~$\bigtheta{\sosdegreestd}$
  but any SOS refutation of~$\fstd_\varnum$
  requires size
  $\varnum^{\bigomega{\sosdegreestd}}$.
\end{theorem}

This theorem extends an analogous result joint by the two authors with
Atserias  in~\cite{ALN14NarrowProofsECCC}
for the proof systems resolution, polynomial calculus, and
Sherali-Adams,\footnote{The exact details of these proof systems are not important
  for this discussion, and so we choose not to elaborate further here,
  instead referring the interested reader to~\cite{ALN14NarrowProofsECCC}.}
where upper bounds on refutation size in terms of width,
degree, and rank, respectively, were shown to be tight up to the
multiplicative constant in the exponent. 
\Refth{th:main-result-informal}
works for all of these proof systems, since the upper bound is in fact
on resolution width (\ie the size of a largest clause in a resolution
refutation), not just SOS degree, and in this sense the 
theorem subsumes the results in~\cite{ALN14NarrowProofsECCC}. The
concrete bound we obtain for the exponent inside the asymptotic
notation in  the~$\varnum^{\bigomega{\sosdegreestd}}$ size lower bound
is very much worse, however, and therefore the gap between upper and lower
bounds is very much larger than in~\cite{ALN14NarrowProofsECCC}.

We want to emphasize
that the size lower bound in 
\refth{th:main-result-informal} holds for 
SOS proofs of arbitrary degree. Thus, going to higher degree (\ie
higher levels of the Lasserre SDP hierarchy) does not help, since even
arbitrarily large degree cannot yield shorter proofs. This is an
interesting parallel to the paper~\cite{LRST14PowerSymmetric}
exhibiting problems for which a (symmetric) SDP relaxation of
arbitrary degree but bounded size $n^{d}$ does not do much better than
the systematic relaxation of degree $d$.

\subsection{Techniques}
\label{sec:intro-techniques}

We obtain the result in \refth{th:main-result-informal} as a special
case of a more general method of amplifying lower bounds on width (in
resolution), degree (in polynomial calculus) and rank/degree (in
Sherali-Adams and Lasserre/SOS) to size lower bounds in the
corresponding proof systems. This method 
is in some sense already implicit
in~\cite{ALN14NarrowProofsECCC}, which in turn relies
heavily on an earlier paper by Atserias \etal~\cite{AMO13LowerBounds},
but it turns out that extracting the essential ingredients and making
them explicit is helpful for extending the results
in~\cite{ALN14NarrowProofsECCC} to an analogue for \sumsofsquares.  
We give a
brief, informal description of the three main ingredients of the method
below.

\ifthenelse{\boolean{conferenceversion}}
{\subparagraph*{(i) Find a base CNF formulas hard \wrt width/degree/rank}}
{\paragraph{(i) Find a base CNF formulas hard \wrt width/degree/rank}}
To start, we need to find a base problem, encoded as an unsatisfiable
CNF formula, that is ``moderately hard'' for the proof system at
hand. What this means is that we should be able to prove
asymptotically tight bounds on width if we are dealing with
resolution, on degree for polynomial calculus, and on degree/rank for
Sherali-Adams and \sumsofsquares.
It then follows by a generic argument (as discussed briefly above for
\SOS) that a bound
$\bigoh{\sosdegreestd}$ 
on width/degree/rank implies an upper bound
$\varnum^{\bigoh{\sosdegreestd}}$
on proof size.

\ifthenelse{\boolean{conferenceversion}}
{In \cite{ALN14NarrowProofsECCC,AMO13LowerBounds} the pigeonhole principle}
{In \cite{AMO13LowerBounds,ALN14NarrowProofsECCC} the pigeonhole principle}
served as the base problem. This principle, which 
has been extensively studied in
proof complexity, is encoded in CNF as 
\introduceterm{pigeonhole principle (PHP) formulas} 
saying that there is a one-to-one mapping of
$m$~pigeons into $n$~pigeonholes for $m>n$.
For \sumsofsquares we cannot use PHP formulas, however, since they are
not hard \wrt SOS~degree. Instead we construct an SOS reduction in low
degree from inconsistent systems of $\F_2$-linear equations to the
clique problem,  and then appeal
to the result
in~\cite{Grigoriev01LinearLowerBound,Schoenebeck08LinearLevel} briefly
discussed above 
to obtain the following degree lower bound.

\begin{theorem}[informal]
  \label{th:degreelowerbound-informal}
  Given 
  $\cliquesize \in \Nplus$, 
  there is a graph
  $\graphstd$ 
  and a $3$-CNF formula
  $\cliqueformula[\cliquesize]{\graphstd}$
  of size polynomial 
  in~$\cliquesize$    
  with the following properties:
  \begin{enumerate}

  \item
    The graph  $\graphstd$ does not contain a 
    $\cliquesize$-clique,
    but
    the formula 
    $\cliqueformula[\cliquesize]{\graphstd}$
    claims that it does.

  \item
    Resolution can refute
    $\cliqueformula[\cliquesize]{\graphstd}$
    in width~$\cliquesize$.

  \item
    Any \sumsofsquares refutation of
    $\cliqueformula[\cliquesize]{\graphstd}$
    requires
    degree~$\bigomega{\cliquesize}$. 
  \end{enumerate}

\end{theorem}

\ifthenelse{\boolean{conferenceversion}}
{\subparagraph*{(ii) Relativize the CNF formulas}}
{\paragraph{(ii) Relativize the CNF formulas}}
The second step is to take the 
formulas for which we have
established width/degree/rank lower bounds and
\introduceterm{relativize}
them.
Relativization is an idea that seems to have been considered for the first 
time in the context of proof complexity by
Kraj{\'\i}{\v{c}}ek~\cite{Krajicek04Combinatorics} and that was further
developed by Dantchev and Riis~\cite{DR03ComplexityGap}.
Very loosely, it can be described as follows.

Suppose that we have a CNF formula encoding (the negation of) a
combinatorial principle saying that some set~$S$ has a property. 
For instance, the CNF formula could encode the pigeonhole principle
discussed above, or could claim the
existence of a totally ordered set
of $n$ elements where no
element in the set is minimal with respect to the ordering
(these latter CNF formulas are known as
\introduceterm{ordering principle formulas}, 
\introduceterm{least number principle formulas}, or 
\introduceterm{graph tautologies} in the literature).

The formula at hand is then relativized by constructing another
formula encoding that there is a (potentially much larger) set~$T$
containing a subset $S \subseteq T$ for which the same
combinatorial principle holds.
For the ordering principle, we can encode that there exists a
non-empty ordered subset
\mbox{$S \subseteq T$} of arbitrary size such that it is possible for all
elements in~$S$ to find a smaller element inside~$S$. 
This relativization step transforms the previously very easy
ordering principle formulas into relativized versions that are
exponentially hard for
resolution~\cite{Dantchev06Relativisation,DM14Relativization}.    
For the PHP formulas, we 
specify
that we have a set of
$M \gg m$ pigeons mapped into into $n < m$~holes such that there
exists a subset of $m$~pigeons 
that are mapped injectively. 

In our setting, it will be important that the relativization does not
make the formulas too hard. We do not want the hardness to blow up
exponentially and instead would like the upper bound obtained in the first
step above to scale nicely with the size of the relativization.  For
our general approach to work, we therefore need formulas talking about some
domain being mapped to some range, where we can enlarge the domain
while keeping the range fixed, and where in addition the mapping is
symmetric in the sense that permuting the domain does not change the
formula.

For this reason, relativizing the ordering principle formulas does not
work for our purposes. Pigeonhole principle formulas have this
structure, however, which is exactly why the proofs
in~\cite{ALN14NarrowProofsECCC} go through. As already mentioned, PHP
formulas will not work for \sumsofsquares, but we can relativize the
formulas in 
\refth{th:degreelowerbound-informal}
by saying that there is a large subset of vertices such that there is
a \mbox{$\cliquesize$-clique} hiding 
inside such a subset.
 
\ifthenelse{\boolean{conferenceversion}}           
{\subparagraph*{(iii) Apply random restrictions to show proof 
    size lower bounds}}
{\paragraph{(iii) Apply random restrictions to show proof 
    size lower bounds}}
In the final step, we use random restrictions to establish lower
bounds on proof size for the relativized CNF formulas
obtained in the second step.
This part of the proof is relatively standard, except for a crucial
twist in the restriction argument introduced
in~\cite{AMO13LowerBounds}.

Assume that there is a small refutation in \sumsofsquares (or whatever
proof system we are studying) of the
relativized formula claiming the existence of a subset of size
$m \ll M$ with the given combinatorial property.
Now hit the formula (and the refutation) with a random restriction
that in effect chooses a subset of size~$m$, and hence
gives us back the original, non-relativized formula.
This restriction will be fairly aggressive in terms of the number of
variables set to fixed truth values, and hence it will hold with high
probability that the restricted refutation has no monomials of high degree
(or, for resolution, no clauses of high width), since all such
monomials will either have been killed by the restriction or at least
have shrunk significantly. (We remark that making use of this
shrinking in the analysis is the crucial extra feature added
in~\cite{AMO13LowerBounds}.)  
But this means that we have a refutation of the original formula in
degree smaller than the lower bound established in the first
step. Hence, no small refutation can exist, and  the lower bound on
proof size follows.

This concludes the overview of our method to amplify lower bounds on
width/degree/rank 
to size.  
It is our hope that
developing such a systematic approach for deriving this kind of lower
bounds, and making explicit what conditions are needed for this
approach to work, can also be useful in other contexts.

\subsection{Organization of This Paper}

The rest of this paper is organized as follows.  We start
in~\refsec{sec:preliminaries} by reviewing the definitions and
notation used, and also stating some basic facts that we will need.
In \refsec{sec:degree-lower-bound}, we prove a degree lower bound for
CNF formulas encoding a version of the clique problem.  We then
present in \refsec{sec:size-lower-bound-relativization} a general
method for obtaining \SOS size lower bounds from degree lower bounds
(or from width, degree, and rank, respectively, for proof systems such
as resolution, polynomial calculus, and Sherali-Adams).  We conclude
with a brief discussion of some possible directions for future
research in \refsec{sec:conclusion}.
\ifthenelse{\boolean{conferenceversion}}
{We refer to the upcoming full version of this paper for the proof
  omitted in this extended abstract.}
{}

\section{Preliminaries}
\label{sec:preliminaries}

For a positive integer~$n$, we use the standard notation
$[n] = \set{1,2, \ldots, n}$.
All logarithms in this paper are 
to base~$2$.
A CNF formula $\formulastd$ is a conjunction of clauses, 
denoted 
$\formulastd=\Land_{\clauseidx}\clc_{\clauseidx}$,
where
each clause $\clc$ is a disjunction of
literals, denoted 
\mbox{$\clc = \Lor_{\litidx} \litstd_{\litidx}$.} 
Each literal~$\litstd$
is either a propositional variable $\varstd$ 
(a \introduceterm{positive literal}) 
or its negation~$\olnot{\varstd}$
(a \introduceterm{negative literal}). 
We think of formulas and clauses as sets, so that there is no
repetition and order does not matter.  
We consider polynomials on the same propositional variables, with the
convention that, as an algebraic variable, $\varstd$ 
evaluates to $1$ when it is true and to $0$ when it is false.
All polynomials in this paper are evaluated on $0$/$1$-assignments, 
and live in the ring of real
multilinear polynomials, which is the ring of real polynomials modulo
the ideal generated by
polynomials
$\varstd_{\indexstd}^{2}-\varstd_{\indexstd}$
for all variables~$\varstd_{\indexstd}$.
In other words, all variables in all monomials have degree at most
one, and 
monomial multiplication 
is defined by
$
\bigl( \prod_{\indexstd\in A} \varstd_{\indexstd}  \bigr)
\cdot  
\bigl( \prod_{\indexstd\in B} \varstd_{\indexstd}  \bigr)
=
\prod_{\indexstd\in A \cup B} \varstd_{\indexstd}
$.

Since \sumsofsquares derivations operate with polynomial equations and
inequalities, in order to reason about CNF formulas we need to encode
them in this language. For a clause
$\clc = \clc^{+} \lor \clc^{-}$, 
where we write 
$\clc^{+}$ and $\clc^{-}$ to denote the subsets of positive and
negative literals, respectively, we define
\begin{equation}
  \label{eq:clause-as-sum}
  \sumencoding{\clc}
  = 
  \sum_{\varstd \in \clc^{+}} \varstd +
  \sum_{\olnot{\varstd} \in \clc^{-}} (1 -\varstd)
\end{equation}
and encode $\clc$ as the inequality 
\begin{equation}
  \label{eq:clause-as-inequality}  
  \sumencoding{\clc} \geq 1
  \eqperiod  
\end{equation}
Clearly, a clause~$\clc$ is satisfied by a $0$/$1$-assignment
if and only if the same assignment satisfies the inequality
$\sumencoding{\clc}\geq 1$. 
For a variable
$\varstd$
and a bit
$\stringstd \in \set{0,1}$,
we define
\begin{equation}
  \indicator{\varstd}{\stringstd}
  \ = \
  \begin{cases}
    1 - \varstd & \text{if $\stringstd = 0$,}
    \\
    \varstd & \text{if $\stringstd = 1$;}
  \end{cases}
\end{equation}
and for a sequence of variables
$\vec{\varstd}=(\varstd_{\indexstd_{1}},\ldots\varstd_{\indexstd_{w}})$
and a binary string $\stringstd=(\stringstd_{1},\ldots
\stringstd_{w})$,
we define the \introduceterm{indicator polynomial} 
\begin{equation}
  \label{eq:indicator-poly}
  \indicatorstd
  \ = \ 
  \prod_{\indexaux = 1}^{w}
  \indicator{\varstd_{\indexstd_\indexaux}}{\stringstd_\indexaux}
\end{equation}
expanded out as a linear combination of monomials.
That is, 
$\indicatorstd$~is the polynomial that  evaluates to~$1$
for $0$/$1$-assignments satisfying the equalities
$\varstd_{\varidx_{\indexaux}}=\stringstd_{\indexaux}$ 
for
$\indexaux = 1 , \ldots, w$
and to~$0$ 
for all other $0$/$1$-assignments.
We have the following useful fact.

\begin{fact}
  \label{fact:indicators_completeness}
  For every sequence of variables $\vec{\varstd}$
  the syntactic equality
  $
  \bigl(
  \sum_{\stringstd\in\{0,1\}^{w}}
  \indicator{\vec{\varstd}}{\stringstd}
  \bigr)
  = 1
  $
  holds (after cancellation of terms).
\end{fact}

Let $\fstd$ be a CNF formula over some set of variables denoted
as~$\varset$, and let 
$\partassign$ be a \introduceterm{partial assignment} 
on~$\varset$.
We write~$\restrict{\formulastd}{\partassign}$ 
to denote the formula~$\fstd$ \introduceterm{restricted by}~$\partassign$,
where all clauses $\clc \in \fstd$ satisfied by~$\partassign$ are
removed and all literals falsified by~$\partassign$ in other clauses
are removed.
For a polynomial~$\polytheorem$ over variables~$\varset$ 
(written, as always, as a linear combination of distinct monomials), 
we let
$\restrict{\polytheorem}{\partassign}$
denote the polynomial obtained by substituting values for assigned
variables and removing monomials that evaluate to~$0$.
We extend this definition to sets of formulas or polynomials in the
obvious way  by taking unions.

\begin{definition}[\Sumsofsquares proof system]
  \label{def:sums-of-squares}
  A 
  \introduceterm{\sumsofsquares derivation},
  or 
  \introduceterm{\SOS derivation}
  for short,
  of the polynomial inequality 
  \mbox{$p \geq 0$} 
  from the system of polynomial constraints
  \begin{equation}
    \label{eq:sos-contraints}
    \polyequation_1=0,\ldots,\polyequation_\polyequationcount=0,
    \ 
    \polyinequality_1\geq0,\ldots,\polyinequality_\polyinequalitycount\geq0
  \end{equation}
        is a sum
  \begin{equation}
    \label{eq:sos-proof}
    \polytheorem
    =
    \sum^{\polyequationcount}_{\clauseidx=1}\polymultiplier_\clauseidx
    \polyequation_\clauseidx  
    + 
    \sum^{\polyinequalitycount}_{\clauseidx=1} 
    \polysos_\clauseidx\polyinequality_\clauseidx
    + 
    \polysos_{0}
    \eqcomma
  \end{equation}
  where 
  $\polymultiplier_{1}, \ldots, \polymultiplier_{\polyequationcount}$
  are arbitrary polynomials
  and each
  $\polysos_\clauseidx$ is
  expressible as a sums of squares 
  $\sum_{\clauseidxaux}\polysoscomp^{2}_{\clauseidx,\clauseidxaux}$.
    A~derivation of the equation $\polytheorem=0$ is a pair of
  derivations of $\polytheorem\geq 0$ and
  $-\polytheorem\geq 0$.  
    A \introduceterm{\sumsofsquares refutation}
  of~\refeq{eq:sos-contraints}
  is a derivation of the inequality $-1 \geq 0$
  from~\refeq{eq:sos-contraints}.

  The \introduceterm{\sosdegree} of an \SOS derivation is the maximum degree
  among all the polynomials 
  $\polymultiplier_{\clauseidx}\polyequation_{\clauseidx}$,
  $\polysos_{\clauseidx}\polyinequality_{\clauseidx}$,
  and~$\polysos_{0}$ 
  in~\eqref{eq:sos-proof}.
    The \introduceterm{\sossize} of an \SOS derivation is the total
  number of  monomials (counted with repetition) in all polynomials
  $\polymultiplier_{\clauseidx}\polyequation_{\clauseidx}$,
  $\polysos_{\clauseidx}\polyinequality_{\clauseidx}$,
  and~$\polysos_{0}$
  (all expanded out as linear combinations of distinct monomials).
    The \sossize and \sosdegree of refuting an unsatisfiable system of
  polynomial constraints
  are defined by taking the minimum over
  all \SOS refutations of the system
  with respect to the corresponding measure.
\end{definition}

\ifthenelse{\boolean{conferenceversion}}{  We remark that our choice of the multilinear setting is without any
  loss of generality and only serves to simplify the technical
  arguments slightly. }{\begin{remark}
  \label{rem:sos-multilinearity}
  Readers more familiar with the usual definition of
  Positivstellensatz/\sumsofsquares in the literature might be a bit
  puzzled by the use of multilinearity in 
  \refdef{def:sums-of-squares}, and might also wonder where
  the axioms
    $\varstd^{2}_{\varidx}-\varstd_{\varidx}=0$,
  $\varstd_{\varidx} \geq 0$,
  and 
  $1-\varstd_{\varidx}\geq 0$ 
  for every variable~$\varstd_{\varidx}$ disappeared.
    It is important to note that we have these axioms in our multilinear
  setting as well,  although they are not explicitly mentioned. 
  Equations of the form  
  $\varstd^{2}_{\varidx}-\varstd_{\varidx}=0$ 
  are tautological due to multilinearity, and the inequalities
  $\varstd_{\varidx} \geq 0$ and $1-\varstd_{\varidx} \geq 0$ 
  are derivable by the squaring rule since
  in the multilinear setting we have
  $\varstd_{\varidx} = \varstd^{2}_{\varidx}$ 
  and
  $1 - \varstd_{\varidx}  = {(1-\varstd_{\varidx})}^{2}$.

  Our choice of 
  the multilinear setting is without any loss of
  generality and only serves to simplify the technical arguments slightly. 
  It is easy to see that applying the multilinearization operator 
  mapping $\varstd^{\genericint}_{\varidx}$ to $\varstd_{\varidx}$ for every 
  $\genericint\geq 1$ to any \SOS derivation over real polynomials yields a legal
  \SOS derivation over multilinear real polynomials in at most the
  same \sossize and \sosdegree.
  Thus, working in the multilinear setting can only make our lower
  bounds stronger. As to the upper bounds in this paper, we prove them
  in the resolution proof system 
  discussed below, 
  and the simulation of
  resolution by \sumsofsquares in \reflem{lmm:simulation} below works
  also in the  standard setting without multilinearization.
\end{remark}
}

Let us state some useful basic
properties of multilinear polynomials for 
\ifthenelse{\boolean{conferenceversion}}
{later reference.}
{later reference
  (and also provide a proof just for completeness).}

\begin{proposition}  [Unique multilinear representation]
  \label{stm:multilinear_representation}
  Every function
  $\funcdescr{f}{\{0,1\}^n}{\R}$ 
  has a unique  representation as a multilinear polynomial.
  In particular,
  if $\polytheorem$
  is a multilinear polynomial such that 
  $\polytheorem(\alpha) \in \set{0,1}$ for all
  $\alpha \in\{0,1\}^{n}$,
  then for every positive integer~$\ell$ the equality
  $\polytheorem^{\ell} = \polytheorem$
  holds
  (where this is a syntactic equality of multlinear polynomials
  expanded out as linear combinations of distinct monomials). 
\end{proposition}

\ifthenelse{\boolean{conferenceversion}}{}{\begin{proof}
  The set of functions from $\{0,1\}^{n}$ to $\R$ is a vector space of
  dimension $2^{n}$. 
    Any function~$f(\vec{\varstd})$ in this space can be
  represented   as   a linear combination
    $\sum_{\stringstd\in\{0,1\}^{n}} f(\stringstd) \cdot
  \indicator{\vec{\varstd}}{\stringstd} (\vec{\varstd}) $.
    Since each
    $\indicator{\vec{\varstd}}{\stringstd} $
  is a multilinear polynomial
  the multilinear monomials on $n$ variables are a set of
  $2^{n}$ generators of the vector space. By linear independence they
  also form a basis, 
  and hence the representation of a function as a linear combination
  of multilinear monomials is unique.
  The second part of the proposition now follows
  immediately since $\polytheorem^{\ell}$ and $\polytheorem$
  compute the same function.
\end{proof}
}

The upper bounds in this paper are shown in the weaker proof system
\introduceterm{resolution},
which is defined as follows.
A \introduceterm{resolution derivation} of a clause $\cld$ from a CNF
formula~$\fstd$ 
is a sequence of clauses
$(\cld_1, \cld_2, \ldots, \cld_\stoptime)$
such that 
$\cld_\stoptime = \cld$
and for every clause $\cld_\indexstd$ it holds that it
is either a clause of~$\fstd$ (an \introduceterm{axiom}),
or is obtained by \introduceterm{weakening} from some
$\cld_\indexaux \subseteq \cld_\indexstd$ for $\indexaux < \indexstd$,
or can be inferred from two clauses
$\cld_\indexauxaux, \cld_\indexaux$, $\indexauxaux < \indexaux < \indexstd$,
by the \introduceterm{resolution rule} that allows to derive the
clause
$\cla \lor \clb$ from two
clauses $\cla \lor \varstd$ and
$\clb \lor \olnot{\varstd}$
(where we say that
$\cla \lor \varstd$ and~$\clb \lor \olnot{\varstd}$
are \introduceterm{resolved on~$\varstd$}
to yield the \introduceterm{resolvent} 
$\cla \lor \clb$).
If in a resolution derivation
$(\cld_1, \cld_2, \ldots, \cld_\stoptime)$
each clause $\cld_\indexaux$
is only used once in a weakening or resolution step to derive some
$\cld_\indexstd$ \mbox{for $\indexstd > \indexaux$}, 
we say that the derivation is \introduceterm{tree-like}
(such derivations may contain multiple copies of the same clause).
A \introduceterm{resolution refutation} of~$\fstd$, or
\introduceterm{resolution proof} for~$\fstd$, is a 
derivation of the empty clause (the clause containing no literals)
from~$\fstd$.

The \introduceterm{\reswidth{}} of a clause is the number of literals
in it, and the 
\reswidth of a CNF formula or resolution derivation is the maximal
\reswidth of any clause in the formula or derivation.
The \introduceterm{\ressize{}} of a resolution derivation is the total
number of clauses in it (counted with repetitions).
The \ressize and \reswidth of refuting an unsatisfiable CNF
formula~$\fstd$ is defined by taking the minimum over all resolution
refutations of~$\fstd$ with respect to the corresponding measure.

The following standard fact is easy to establish by forward induction
over resolution derivations. We omit the proof.

\begin{fact}
  \label{fact:proofweakening}
  Consider a partial assignment $\partassign$ which assigns $\genericint$
  variables. 
  Let $\cla$ be the unique clause of \reswidth~$\genericint$ such
  that 
  $\cla$ evaluates to false under~$\partassign$.
                  If 
  resolution can derive
  $\clc$ in \reswidth~$\reswidthstd$ and
  \ressize~$\ressizestd$ from~$\restrict{\formulastd}{\partassign}$,
  then resolution can derive $\cla \lor \clc$ in \reswidth
  at most~$\reswidthstd+\ell$ and 
  \ressize at most~$\ressizestd+1$ from~$\formulastd$.
                  \end{fact}

Let us also state for the record the formal claim that \SOS is more
powerful than resolution in term of \sosdegree (and for constant
\sosdegree also in terms of \sossize). 
The next lemma is essentially Lemma~4.6
in~\cite{ALN14NarrowProofsECCC}, except that there the lemma is stated for
the Sherali-Adams proof system. Since \SOS simulates Sherali-Adams
efficiently \wrt both \sossize and \sosdegree, however, the same
bounds apply also for \SOS.
\ifthenelse{\boolean{conferenceversion}}{}{Referring to the discussion in 
\refrem{rem:sos-multilinearity}, it
should also be pointed out that the lemma
in~\cite{ALN14NarrowProofsECCC} is proven in the  more common
non-multilinear setting with explicit axioms
$\varstd^{2}_{\varidx}-\varstd_{\varidx}=0$,
$\varstd_{\varidx} \geq 0$,
and 
$1-\varstd_{\varidx}\geq 0$ 
for all variables $\varstd_{\varidx}$.
}

\begin{lemma}[SOS simulation of resolution]
  \label{lmm:simulation}
  If a CNF formula
  $\formulastd=\Land^{\clausecount}_{\clauseidx=1}\clc_{\clauseidx}$
  has a 
  resolution refutation of \ressize $\ressizestd$ and \reswidth
  $\reswidthstd$,
    then the constraints $\{\sumencoding{\clc_{\clauseidx}}\geq
  1\}^{\clausecount}_{\clauseidx=1}$
  as defined in
  \refeq{eq:clause-as-sum}
  and
  \refeq{eq:clause-as-inequality}  
    have an \SOS refutation of \sossize 
    $\Bigoh{\reswidthstd 2^{\reswidthstd}\ressizestd}$ and 
    \sosdegree
  at most
  $\reswidthstd+1$.
\end{lemma}

\ifthenelse{\boolean{conferenceversion}}
{The next lemma will be useful as a subroutine when we prove upper
  bounds in resolution. We again omit the proof.}
{The next lemma will be useful as a subroutine when we prove upper
  bounds in resolution.}

\begin{lemma}
  \label{lmm:bruteforce}
  Let
  $\domainsmall$
  and
  $m_{1},m_{2},\ldots m_{\domainsmall}$
  be positive numbers. 
  Then the CNF formula consisting of the clauses
    \begin{subequations}
    \begin{align}
      &
      \thrvaraux_{\indexstd,0} 
      && \text{$\indexstd\in [\domainsmall]$,}
      \\
      &
      \olnot{\thrvaraux}_{\indexstd,\indexaux-1} \lor
      \varstd_{\indexstd,\indexaux} \lor
      \thrvaraux_{\indexstd,\indexaux} 
      &&
      \text{$\indexstd\in [\domainsmall]$,
        $\indexaux\in[m_{\indexstd}]$,}
      \\
      &
      \olnot{\thrvaraux}_{\indexstd,m_{\indexstd}} 
      && 
      \text{$\indexstd\in [\domainsmall]$,}
      \\
      &
      \olnot{\varstd}_{1,\indexaux_{1}} \lor
      \olnot{\varstd}_{2,\indexaux_{2}} \cdots \lor 
      \olnot{\varstd}_{\domainsmall,\indexaux_{\domainsmall}} 
      &&
      \text{$(\indexaux_{1},\ldots,\indexaux_{\domainsmall})
        \in [m_{1}] \times\cdots \times[m_{\domainsmall}]$,}
    \end{align}
  \end{subequations}
  has a resolution refutation of \reswidth $\domainsmall+1$
  and \ressize $\Bigoh{\prod_{\indexstd=1}^{\domainsmall}m_{\indexstd}}$.
\end{lemma}

\ifthenelse{\boolean{conferenceversion}}{}{\begin{proof}
  We prove the lemma by backwards induction over~$\domainsmall$.
  Consider any clause~$\cla$ of the form 
  \begin{equation}
    \label{eq:backwardinduction}
    \cla=
    \olnot{\varstd}_{1,\indexaux_{1}} \lor
    \olnot{\varstd}_{2,\indexaux_{2}} \cdots \lor 
    \olnot{\varstd}_{(\indexstd-1),\indexaux_{(\indexstd-1)}}
  \end{equation}
  for $1 \leq \indexstd\leq\domainsmall$
  (and note that for $\indexstd = 1$ this is the empty clause).
    We will show how to
  derive~$\cla$ in \reswidth $\indexstd+1$
    given clauses
    $
  \cla \lor \olnot{\varstd}_{i,1}, \, 
  \cla \lor \olnot{\varstd}_{i,2}, \,
  \ldots, \,
  \cla \lor \olnot{\varstd}_{i,m_{\indexstd}}
  $.

  We start by resolving the axioms
  $\thrvaraux_{\indexstd,0}$ 
  and 
  $\olnot{\thrvaraux}_{\indexstd,0} \lor
  \varstd_{\indexstd,1} \lor
  \thrvaraux_{\indexstd,1}$, 
  and then we apply the resolution rule again on this resolvent 
  and the clause 
  $\cla \lor \olnot{\varstd}_{\indexstd,1}$ 
  (available by the induction hypothesis)
  to get
  $\cla \lor \thrvaraux_{\indexstd,1}$.
  We now deduce
  $\cla \lor \thrvaraux_{\indexstd,\indexaux}$ for increasing $\indexaux$. 
    Suppose we have already obtained
  $\cla \lor \thrvaraux_{\indexstd,\indexaux-1}$.
  Using the inductively derived clause
  $\cla \lor \olnot{\varstd}_{\indexstd,\indexaux}$
  and the axiom
  $
  \olnot{\thrvaraux}_{\indexstd,\indexaux-1} \lor
  \varstd_{\indexstd,\indexaux} \lor
  \thrvaraux_{\indexstd,\indexaux}
  $,
  we can resolve   on variables 
  $\thrvaraux_{\indexstd,\indexaux-1}$
  and 
  $\varstd_{\indexstd,\indexaux}$
  to obtain
  $\cla \lor \thrvaraux_{\indexstd,\indexaux}$.
  Once $\cla \lor \thrvaraux_{\indexstd,m_\indexstd}$ 
  has been derived, we resolve it with the axiom
  $\olnot{\thrvaraux}_{\indexstd,m_\indexstd}$ to get $\cla$.
    By backward induction we reach  the empty clause for 
  $\indexstd=1$, which concludes the resolution refutation.
    Since $\indexstd \leq k$, the refutation has width $\domainsmall+1$.
    It is easy to verify that all axioms and intermediate clauses in the
  refutation are used exactly once. Thus, the refutation is tree-like,
  and has size exactly twice the number of axioms clauses minus one,
  which, in particular, is 
  $\Bigoh{\prod_{\indexstd=1}^{\domainsmall}m_{\indexstd}}$.
\end{proof}

}
When we construct formulas to be relativized as described in
\refsec{sec:intro-techniques}, it is convenient to use variables
$\varstd_{\indexstd,\myvec{\indexaux}}$, 
 where $\indexstd$ ranges over some
specific domain~$\domainset$ and
$\myvec{\indexaux}\,$ is a collection of other indices.
We say that the variable
$\varstd_{\indexstd,\myvec{\indexaux}}$
\introduceterm{mentions} the element $\indexstd \in \domainset$.
The \introduceterm{\indexwidth} of a clause is the number of distinct elements
of $\domainset$ mentioned by its variables.
The \indexwidth of a CNF formula or resolution proof is defined by
taking the maximum \indexwidth over all 
its
clauses, 
and the \indexwidth of refuting a CNF formula~$\fstd$ is the minimal
\indexwidth of any resolution refutation of~$\fstd$.
Similarly, the \introduceterm{\indexdegree} of a monomial is the
number of distinct elements in~$\domainset$ mentioned by its
variables, the \indexdegree of a polynomial or \SOS proof is the
maximal \indexdegree of any monomial in it, and the \indexdegree of
refuting an unsatisfiable system of polynomial constraints is
defined by taking the minimum over all refutations.

\section{A Degree Lower Bound for Clique Formulas}
\label{sec:degree-lower-bound}

In this section we state and prove the formal version of
\refth{th:degreelowerbound-informal}, namely a lower bound for the
\indexdegree needed in \SOS to prove that a graph $\graphstd$ has no
$\cliquesize$-clique.  Let us start by describing how we encode the
$k$-clique problem as a CNF formula.

\begin{definition}
  [$\cliquesize$-clique formula]
  \label{def:cliqueformula}
  Let $\cliquesize$ be a positive integer, 
  $\graphstd = (V, E)$ be an undirected graph on $\cliquevertexnum$~vertices, 
  and 
  $(v_{1},v_{2},\ldots,v_{\cliquevertexnum})$
  be an enumeration 
  of~$V(\graphstd) = V$.
        Then the formula \cliqueformula{\graphstd} 
  consists of the clauses
  \begin{subequations}
    \begin{align}
      &
        \olnot{\cliquevar}_{\domainidx,u} \lor
        \olnot{\cliquevar}_{\domainidx',v} 
      &&
         \text{           $\domainidx,\domainidx' \in [\cliquesize]$,
           $\domainidx \neq \domainidx'$,
           $\{u,v\} \not\in E(\graphstd)$,
         }\label{eq:edge_clause}
      \\
      &
        \olnot{\cliquevar}_{\domainidx,u} \lor
        \olnot{\cliquevar}_{\domainidx,v},
      &&
         \text{$\domainidx \in [\cliquesize]$,
           $u,v \in V(\graphstd)$,
           $u \neq v$,
         }\label{eq:functional_clause}
      \\
      &
        \cliqueaux_{\domainidx,0} 
      &&
         \text{$\domainidx\in[\cliquesize]$,
         }\label{eq:index_clause_cnfA}         
      \\
      &
        \olnot{\cliqueaux}_{\domainidx,{(\indexaux-1})} \lor 
        \cliquevar_{\domainidx,v_\indexaux} \lor
        \cliqueaux_{\domainidx,\indexaux}
      &&
         \text{$\domainidx\in[\cliquesize]$,
         $\indexaux\in[\cliquevertexnum]$,
         }\label{eq:index_clause_cnfB}
      \\
      &
        \olnot{\cliqueaux}_{\domainidx,\cliquevertexnum}
      &&
         \text{$\domainidx\in[\cliquesize]$.
         }\label{eq:index_clause_cnfC}
    \end{align}
  \end{subequations}
\end{definition}

The formula \cliqueformula{\graphstd} encodes the claim that
$\graphstd$ has a clique of size $\cliquesize$.
The intended meaning of 
the
variable $\cliquevar_{\domainidx,v}$ for
$v \in V(\graphstd)$ 
is that $v$ is 
the
$\domainidx$th 
vertex 
of the clique.
\ifthenelse{\boolean{conferenceversion}}{}{The clauses in~\eqref{eq:edge_clause} enforce that any two members of
the clique are distinct and are connected by an edge.
The clauses in~\eqref{eq:functional_clause} enforce that at most one
vertex is chosen for each $\domainidx \in [\cliquesize]$.
The clauses
in~\eqref{eq:index_clause_cnfA}--\eqref{eq:index_clause_cnfC}
are simply the $3$-CNF encoding (using extension variables) of the clause
$\Lor^{\cliquevertexnum}_{\indexaux=1}
\cliquevar_{\domainidx,v_{\indexaux}}$
enforcing that at least one vertex is chosen for each
$\domainidx \in [\cliquesize]$.
}
The variables of \cliqueformula{\graphstd} are indexed by $\domainidx$ 
over the domain $[\cliquesize]$ and the \indexwidth of the formula is~$2$.
The next
proposition shows that 
the naive
brute-force approach to decide
$\cliqueformula{\graphstd}$ can be carried on in resolution
(and hence by \reflem{lmm:simulation} also in \SOS).

\begin{proposition}
  \label{stm:cliqueupperbound}
  If $\graphstd$ has no clique of size $\cliquesize$, then
  \cliqueformula{\graphstd} has a resolution refutation of
  \ressize
  $\Bigoh{\setsize{V}^{\cliquesize}}$
  and
  \reswidth $\cliquesize+1$.  
\end{proposition}

\begin{proof}
  We first use the weakening rule to derive all clauses of the form
  \begin{equation}
    \label{eq:clique_sequence}
    \olnot{\cliquevar}_{1,u_1} \lor \olnot{\cliquevar}_{2,u_2} \lor
    \cdots \lor \olnot{\cliquevar}_{\cliquesize,u_\cliquesize} 
  \end{equation}
  for every sequence of vertices
  $(u_{1},u_{2},\ldots,u_{\cliquesize})$.
  This is possible since either the sequence contains a
  repetition or it includes two vertices with no edge between them,
  and in both cases this means that the
  clause~\refeq{eq:clique_sequence} is a superclause of some 
  clause of the form~\eqref{eq:edge_clause}. 
    Then we derive the empty clause by applying \reflem{lmm:bruteforce}
  to the
  clauses~\eqref{eq:index_clause_cnfA}--\eqref{eq:index_clause_cnfC}
  and~\eqref{eq:clique_sequence}. 
\end{proof}

In order to obtain suitably hard instances of
$\cliqueformula{\graphstd}$ we construct a reduction from
\xorformula{}s to $k$-partite graphs. It is convenient for us to
describe the special case of $\cliquesize$-clique on
$\cliquesize$-partite graphs directly as an encoding as polynomial
equations and inequalities as follows next.

\begin{definition}  [Polynomial encoding of $\cliquesize$-clique on $\cliquesize$-partite graphs]
  \label{def:blockformula}
  For a $\cliquesize$-partite graph $\graphstd$ with
  $V(\graphstd) = 
  V_{1}\disjointunion V_{2} \disjointunion \cdots \disjointunion
  V_{\cliquesize}$ 
  we let 
    \blockformula{\graphstd} denotes the following collection of
  polynomial constraints:
  \begin{subequations}
    \begin{align}
      &
        \sum_{v \in V_{\domainidx}} \cliquevar_{v}= 1
      && \text{$\domainidx \in [\cliquesize]$,}
         \label{eq:blockindex_clause}
      \\
      &
      \label{eq:block-edge}
        \cliquevar_{u} + \cliquevar_{v} \leq 1
      && \text{$u \in V_\domainidx, v \in V_{\domainidx'}$, 
          $\domainidx\neq\domainidx'$, $\{u,v\}\not\in \edges{\graphstd}$.}
    \end{align}
  \end{subequations}
\end{definition}

\ifthenelse{\boolean{conferenceversion}}{}{It is straightforward to verify that these constrants encode the claim
that $\graphstd$ has a clique with one element in each block
$V_{\domainidx}$, since exactly one element is chosen from each block
by~\refeq{eq:blockindex_clause} and all the chosen elements have to be
pairwise connected by~\refeq{eq:block-edge}.

Any lower bound on degree that we establish for
$\blockformula{\graphstd}$ will hold also for
$\cliqueformula{\graphstd}$ as stated in the following proposition.
}

\begin{proposition}
  \label{stm:indexdegreeblock}
  Consider a $\cliquesize$-partite graph $\graphstd$, where 
  $V(\graphstd)=V_{1}\disjointunion V_{2} \disjointunion \cdots
  \disjointunion V_{\cliquesize}$. 
    If \cliqueformula{\graphstd} has an \SOS refutation in \indexdegree
  $\sosdegreestd$, then \blockformula{\graphstd} has an \SOS
  refutation in \indexdegree $\sosdegreestd$.
\end{proposition}

\begin{proof}
  The proof is by transforming
  a refutation of~\cliqueformula{\graphstd} 
  into a refutation of \blockformula{\graphstd} of the
  same \indexdegree.
    To give an overview, we
  start with a refutation of \cliqueformula{\graphstd} of
  \indexdegree~$\sosdegreestd$ and 
  replace
  its variables with
  polynomials of degree at most~$1$ mentioning only variables from
  \blockformula{\graphstd}.
    In this way we get an \SOS refutation of \indexdegree at most
  $\sosdegreestd$ from the substituted axioms of
  \cliqueformula{\graphstd}.
    The latter polynomials are not necessarily axioms of
  \blockformula{\graphstd}, 
  but we show that they have 
  \SOS derivations
  of \indexdegree~$1$ from the axioms of
  \blockformula{\graphstd}. This concludes the proof.

  The variable substitution has two steps: first we substitute every variable
  $\cliqueaux_{\domainidx,\indexaux}$ with the linear form
  $\sum^{\cliquevertexnum}_{t=\indexaux+1}
  \cliquevar_{\domainidx,v_t}$,
  where $\{v_{\indexaux}\}^{\cliquevertexnum}_{\indexaux=1}$ is the
  enumeration of $V(\graphstd)$ in 
  \refdef{def:cliqueformula}, and
  then we set $\cliquevar_{\domainidx,v_{\indexaux}}$ to $0$
  whenever $v_{\indexaux}\not\in V_{\domainidx}$.

  As mentioned above, we now need to give  
  \SOS derivations of   \indexdegree~$1$ of all
  transformed axioms in \cliqueformula{\graphstd}
  from~\blockformula{\graphstd}.
    For the axioms~\eqref{eq:index_clause_cnfA}--\eqref{eq:index_clause_cnfC}, 
  the \SOS encoding is
  \begin{subequations}
    \begin{align}
            &
        \cliqueaux_{\domainidx,0} \geq 1
      &&
        \domainidx\in[\cliquesize],
         \label{eq:index_clauseA}
      \\
      &
      \bigl( 1 - \cliqueaux_{\domainidx,(\indexaux-1)} \bigr) 
      + \cliquevar_{\domainidx,v_\indexaux} 
      + \cliqueaux_{\domainidx,\indexaux} \geq 1 
      &&
         \domainidx\in[\cliquesize], \indexaux\in[\cliquevertexnum],
         \label{eq:index_clauseB}
      \\
      &
      (1 - \cliqueaux_{\domainidx,\cliquevertexnum}) \geq 1
      &&
         \domainidx\in[\cliquesize].
         \label{eq:index_clauseC}
    \end{align}
  \end{subequations}
    
  After the first step of the
  substitution the
  inequalities~\eqref{eq:index_clauseA},~\eqref{eq:index_clauseB} 
  and~\eqref{eq:index_clauseC} become, respectively, 
    the inequality 
  $\sum^{\cliquevertexnum}_{\indexaux=1}
  \cliquevar_{\domainidx,v_{\indexaux}} \geq 1$,
    and 
  two occurrences of tautology $1 \geq 1$.
    Furthermore, after the second step of the substitution the
  inequality~\eqref{eq:index_clauseA} becomes
  $\sum_{v \in V_{\domainidx}} \cliquevar_{\domainidx,v} \geq 1$,
    which is subsumed by
  Equation~\eqref{eq:blockindex_clause}.
  Each of the axioms~\eqref{eq:edge_clause} and~\eqref{eq:functional_clause} 
  is encoded as 
  \begin{equation}
    \label{eq:subst-ineq}
    1 - \cliquevar_{\domainidx,u} -
    \cliquevar_{\domainidx',v} \geq 0     
  \end{equation}  
  for some pair of indices $\domainidx,\domainidx'$ and vertices~$u,v$.
    We assume that $u \in V_{\domainidx}$ and $v \in V_{\domainidx'}$,
  because otherwise the variable substitution 
  turns the inequality  into
  either a tautology or 
  into
  $1 - \cliquevar_{\domainidx,u} \geq 0$,
  where the latter 
  follows from
  ${(1 - \cliquevar_{\domainidx,u} )}^{2}\geq 0$ 
  by multilinearity.
  If $\domainidx\neq\domainidx'$ then the inequality 
  \refeq{eq:subst-ineq}
  is an axiom of 
 \blockformula{\graphstd}. If that is not the case, then we can 
 obtain
 $1 - \varstd_{\domainidx,u} - \varstd_{\domainidx,v}$
 in \indexdegree~$1$ using the derivation
  \begin{equation}
    \underbrace{1 - \sum_{v\in V_\domainidx}
      \varstd_{\domainidx,w}}_{\text{from
        Equation~\eqref{eq:blockindex_clause}}} 
    +
    \underbrace{\sum_{w\not\in\{u,v\}}
      {(\varstd_{\domainidx,w})}^{2}}_{{\text{sum of squares}}}
    =
    1 - \sum_{v\in V_\domainidx} \varstd_{\domainidx,w} 
    +
    \sum_{w\not\in\{u,v\}}\varstd_{\domainidx,w}
    = 1 - \varstd_{\domainidx,u} - \varstd_{\domainidx,v}
  \end{equation}
  where the first identity holds by multilinearity.
  The proposition follows.
                \end{proof}

What we want to do now is to prove a \indexdegree lower bound for
instances of \blockformula{\graphstd} 
where the graph~$\graphstd$ is obtained by a reduction from
(unsatisfiable) sets of \mbox{$\F_2$-linear} equations.
We rely on the version of Grigoriev's
\sosdegree lower
bound~\cite{Grigoriev01LinearLowerBound}
shown by Schoenebeck~\cite{Schoenebeck08LinearLevel},
which is conveniently stated for random $3$-XOR formulas as encoded
next.

\begin{definition}
  [Polynomial encoding of random \xorformula]
  \label{def:randomxor}
  A random \xorformula formula $\xorstd$ represents a system of
  $\Delta \xorvarnum$ 
  linear equations modulo~$2$ defined over $\xorvarnum$
  variables.
    Each equation is sampled at random among all equations of the form
  $\xorvarA \oplus \xorvarB \oplus \xorvarC = b$ as follows:
  $\xorvarA$, $\xorvarB$, $\xorvarC$ are sampled uniformily without
  replacement from the set of $\xorvarnum$ variables and $b$ is
  sampled uniformly in $\set{0,1}$.
    The polynomial encoding of any such linear equation modulo~$2$ 
  is
  \begin{subequations}
    \label{eq:xor_clauses}
    \begin{align}
      (1 - \xorvarA)   (1-\xorvarB)  \xorvarC     &=0\label{eq:xor_clauses_first}\\
      (1 - \xorvarA)   \xorvarB      (1-\xorvarC) &=0\\
      \xorvarA         (1-\xorvarB)  (1-\xorvarC) &=0\\
      \xorvarA         \xorvarB      \xorvarC     &=0\\
      \intertext{when $b = 0$ and}
      (1 - \xorvarA)  (1-\xorvarB)  (1-\xorvarC) &=0\\
      \xorvarA        \xorvarB      (1-\xorvarC) &=0\\
      \xorvarA        (1-\xorvarB)  \xorvarC     &=0\\
      (1 - \xorvarA)  \xorvarB      \xorvarC     &=0\label{eq:xor_clauses_last}
    \end{align}
  \end{subequations}
  when $b=1$.
\end{definition}

Fixing $\delta=1/4$ and $\Delta=8$ in \cite{Schoenebeck08LinearLevel}
we have the following theorem.

\begin{theorem}[\cite{Schoenebeck08LinearLevel}]
  \label{thm:3XOR}
  There exists an $\xordegreelb$, $0 <  \xordegreelb < 1$,
  such that for  every $\epsilon > 0$ 
  there exists an $\xorminvar\in \N$ such that a random
  \xorformula formula $\xorstd$ in $\xorvarnum\geq \xorminvar$
  variables and $8\xorvarnum$ constraints has the following properties
  with probability at least $1 - \epsilon$.
  \begin{enumerate}
  \item At most $6\xorvarnum$ parity constraints of $\xorstd$ can be
    simultaneously satisfied.
  \item Any \sumsofsquares refutation of $\xorstd$ requires degree
    $\xordegreelb \xorvarnum$.
  \end{enumerate}
\end{theorem}

Now we are ready to describe how to transform a
\xorformula formula $\xorstd$ into a 
$\cliquesize$-partite graph~$\graphclausestd$ 
that has a clique of size $\cliquesize$ if and only
if $\xorstd$ is satisfiable.

\begin{definition}  [\xorformula graph]
  \label{def:reduction}
  Given 
  $\cliquesize \in \N$
  and
  a \xorformula formula $\xorstd$ with $8\xorvarnum$~constraints
  over $\xorvarnum$~variables,
  where we assume for simplicity that
  $\cliquesize$ divides
  $8\xorvarnum$,
  we construct a 
  \introduceterm{\xorformula graph $\graphclausestd$}
  as follows.

  We arbitrarily split the formula $\xorstd$ into $\cliquesize$ linear
  systems with ${8\xorvarnum} / {\cliquesize}$ 
  constraints each, denoted as 
  $\xorstd_{1}, \xorstd_{2}, \ldots \xorstd_{\cliquesize}$.
  For each $\xorstd_{\domainidx}$ we 
      let
  $V_{\domainidx}$ 
  be a set of at most 
  $\cliquevertexnum \leq  2^{{24\xorvarnum} / {\cliquesize}}$
  vertices labelled by all possible assignments to the
  at most   ${24\xorvarnum} / {\cliquesize}$ variables
  appearing in $\xorstd_{\domainidx}$.
  For two distinct vertices $u\in V_{\domainidx}$ and
  $v \in V_{\domainidx'}$
  there is an edge between $u$ and $v$ in~$\graphclausestd$ if
  the two assignments corresponding to $u$ and $v$ are compatible, \ie
  when  they assign the same values to the common variables, 
  and also the union of the two assignments 
  does  not violate any constraint in~$\xorstd$.
    (In particular, each $V_{\domainidx}$ is an independent set, since
  two distinct assignments to the same set of variables are not
  compatible.) 
\end{definition}

The key property of the reduction in \refdef{def:reduction} is that it
allows small \indexdegree refutations of
\Blockformula{\graphclausestd} to be converted into small \sosdegree
refutations of~$\xorstd$.

\begin{lemma}
  \label{lmm:reduction}
  If \Blockformula{\graphclausestd} has an \SOS refutation of
  \indexdegree $\sosdegreestd$, then $\xorstd$ has an \SOS
  refutation of \sosdegree
  ${24\sosdegreestd\xorvarnum} / {\cliquesize}$.
  \end{lemma}

\begin{proof}
  Again we start by giving an overview of the proof, which works by
  transforming a refutation of
  \Blockformula{\graphclausestd} of \indexdegree~$\sosdegreestd$
  into a refutation of $\xorstd$ of \sosdegree
    ${24\sosdegreestd\xorvarnum} / {\cliquesize}$. 
  
  Given a refutation 
  of \Blockformula{\graphclausestd} of
  \indexdegree $\sosdegreestd$,
  we
      replace every variable $\cliquevar_{v}$ with a polynomial
  over the variables of $\xorstd$.
  In this way we get an \SOS refutation from the polynomials
  corresponding to the substituted axioms of \Blockformula{\graphclausestd}.
    The latter polynomials need not be axioms of~$\xorstd$, 
    but we show that they can be efficiently derived in \SOS
  from~$\xorstd$. 
    We thus obtain
  an \SOS refutation of~$\xorstd$, the
  \sosdegree of which is easily verified to be as in the statement of
  the lemma.

  We now describe the substitution in 
    detail.
  Consider a block
  $V_{\domainidx}$
  and suppose that the corresponding
  \xorformula formula~$\xorstd_{\domainidx}$ mentions
  $\blockvarnum$~variables. Let us write $\myvec{\xorvar}$ to denote
  this set of variables.
    Then every 
  vertex $v \in V_{\domainidx}$ represents an assignment
  $\stringstd\in\{0,1\}^{\blockvarnum}$ 
  to~$\myvec{\xorvar}$.
    In what follows, we denote the indicator
  polynomial~$\indicator{\myvec{\xorvar}}{\stringstd}$ 
  in~\refeq{eq:indicator-poly}
  by 
  $\indicatorbase_{v}$ for
  brevity, and we substitute for each variable $\cliquevar_{v}$ the
  polynomial
  $\indicatorbase_{v}$ of degree
  $t \leq {24\xorvarnum} / {\cliquesize}$.

  Before the substitution
  each monomial in the original refutation
  has \indexdegree at most $\sosdegreestd$ 
  by assumption. 
    Two important observations are that 
  ${(\indicatorbase_{v})}^{2} = \indicatorbase_{v}$ 
  for
  every $v \in V_{\domainidx}$ and
  that
  $ \indicatorbase_{u}\indicatorbase_{v} =0 $ for every two distinct
  $u,v$ in the same block~$V_{\domainidx}$. Therefore, 
  after the substitution
  each monomial
  is either identically zero or the product of at most $\sosdegreestd$
  indicator polynomials, 
  and hence
  its \sosdegree is at most
  $ {24\sosdegreestd\xorvarnum} / {\cliquesize}$.
    \ifthenelse{\boolean{conferenceversion}}           
  {}
  {To verify these observations, note that the
  identity ${(\indicatorbase_{v})}^{2} = \indicatorbase_{v} $
  holds 
  by Proposition~\ref{stm:multilinear_representation}.
      The equality
  $\indicatorbase_{u}\indicatorbase_{v}=0$ holds because 
  $\indicatorbase_{u}$ and~$\indicatorbase_{v}$
  are the indicator polynomials
  of two incompatible assignments, 
      and so their product always evaluates to zero. 
  Applying
  Proposition~\ref{stm:multilinear_representation} 
  again, we conclude that
  the 
  (multilinear)
  polynomial
  $\indicatorbase_{u}\indicatorbase_{v}$ is identically zero.}

  In order to complete the proof outline above, we now need to 
  present
  \SOS derivations 
  starting
  from the 
  \xorformula constraints of~$\xorstd$ of all polynomial
  constraints resulting from the substitutions in the axioms
  of~\Blockformula{\graphclausestd} described above, and to do so in
  \sosdegree at \mbox{most ${24\xorvarnum} / {\cliquesize}$.}

  Let us first look at the 
  axioms~\eqref{eq:blockindex_clause}. 
    By Fact~\ref{fact:indicators_completeness}, the  identity
  \begin{equation}
    \sum_{v \in V_{\domainidx}}
    \indicatorbase_{v} \ =
    \sum_{\stringstd\in\{0,1\}^{\blockvarnum}}
    \indicator{\myvec{\xorvar}}{\stringstd}=1
  \end{equation}
  holds syntactically,
  so substitutions in axioms of the form~\eqref{eq:blockindex_clause}
  result in tautologies~$1=1$.

  The remaining axioms of \Blockformula{\graphclausestd} 
  in~\refeq{eq:block-edge} 
  have the form
  $\cliquevar_{u}+\cliquevar_{v} \leq 1$
  for non-edges $(u,v)$ between vertices in different blocks.
    By construction of $\graphclausestd$ 
  the reason $u$ and~$v$ are not connected is
  either 
  that
  the partial assignments corresponding to the two vertices
  are incompatible, or 
  that
  their union violates
  some constraint
  in
  $\xorstd$.

  In the first case,
  $1 - \indicatorbase_{u} - \indicatorbase_{v} \geq 0$ is an \SOS
  axiom because of the identity
    \begin{equation}\label{eq:simulate_axioms}
  {(1-\indicatorbase_{u}-\indicatorbase_{v})}^{2} =
    1 -\indicatorbase_{u}-\indicatorbase_{v}
  \eqcomma
  \end{equation}
    which
  follows from the observation
  that $\indicatorbase_{u}$ and~$\indicatorbase_{v}$ are the indicator
  polynomials of two incompatible assignments and
  cannot evaluate
  to~$1$ simultaneously, 
  and
  so
    $(1-\indicatorbase_{u}-\indicatorbase_{v})$ evaluates to either $0$
  or $1$ and is identical to its square by
  Proposition~\ref{stm:multilinear_representation}. The degree
  of~\eqref{eq:simulate_axioms} is ${24\xorvarnum} / {\cliquesize}$.

  In the second case, the two assignments corresponding to $u$ and~$v$ are
  compatible but their union violates some  
  initial equation
  $\polyequation=0$
  of the
  form~\eqref{eq:xor_clauses_first}--\eqref{eq:xor_clauses_last}.
  Any such $f$ is a \mbox{\sosdegree-$3$} indicator polynomial
  which
    evaluates to~$1$ 
  whenever the assignment satisfies the equations
  $\indicatorbase_{u}\indicatorbase_{v}=1$.
  This means that $\indicatorbase_{u}\indicatorbase_{v}$
  contains~$\polyequation$ as a factor.
  We factorize
  $\polyequation$ as $\polyequation_{u}\polyequation_{v}$ 
  so that 
  $\indicatorbase_{u} = \polyequation_{u}\indicatorbase'_{u}$ 
  and
  $\indicatorbase_{v} = \polyequation_{v}\indicatorbase'_{v}$.
    Given this notation, we can derive
  $0 \leq 1-\indicatorbase_{u}-\indicatorbase_{v}$ 
  using the indentity
  \begin{equation}
    \label{eq:second-case-block-xor-bound}
  {(1-\polyequation_{u}-\polyequation_{v})}^{2} 
  +
  {(\polyequation_{u}- \indicatorbase_{u})}^{2} 
  +
  {(\polyequation_{v}- \indicatorbase_{v})}^{2} 
      - 2\polyequation_{u}\polyequation_{v} =
    1 -\indicatorbase_{u}-\indicatorbase_{v}
  \end{equation}
  of degree at most ${24\xorvarnum} / {\cliquesize}$. 
    \ifthenelse{\boolean{conferenceversion}}           
  {}
  {To verify~\eqref{eq:second-case-block-xor-bound},
  observe that the 
  left-hand side  is the
  sum of some squared polynomials and
  $ -2 \polyequation_{u}\polyequation_{v}
  = -2\polyequation = 0$.
    Expanding the squared polynomials and using
  \refpr{stm:multilinear_representation} repeatedly 
  we have that
  $(\polyequation_{u})^{2} = \polyequation_{u}$,
  $(\polyequation_{v})^{2} = \polyequation_{v}$, 
  $(\indicatorbase_{u})^{2} = \indicatorbase_{u}$,
  and 
  $(\indicatorbase_{v})^{2} = \indicatorbase_{v}$, 
  from which we also conclude that
  \begin{equation}
    \polyequation_{u}\indicatorbase_{u} =
    \polyequation_{u}\bigl(\polyequation_{u}\indicatorbase'_{u}\bigr) =
    {\bigl(\polyequation_{u}\bigr)}^{2}\indicatorbase'_{u} =
    \polyequation_{u}\indicatorbase'_{u} = \indicatorbase_{u}
  \end{equation}
  and
  \begin{equation}
    \polyequation_{v}\indicatorbase_{v} =
    \polyequation_{v}\bigl(\polyequation_{v}\indicatorbase'_{v}\bigr) =
    {\bigl(\polyequation_{v}\bigr)}^{2}\indicatorbase'_{v} =
    \polyequation_{v}\indicatorbase'_{v} = \indicatorbase_{v}
  \end{equation}
  which establishes that~\eqref{eq:second-case-block-xor-bound} holds.}
  The lemma follows.
\end{proof}

Now we can put together all the material in this section to prove a
formal version of
\refth{th:degreelowerbound-informal} 
as stated next.

\begin{theorem}
  \label{thm:degreelowerbound}
  There are universal constants 
  $\cliquevertexbound\in\Nplus$ 
  and 
  $\cliquedegreebound$, 
  $0 < \cliquedegreebound < 1$,
  such that 
  for every $\cliquesize \geq 1$ there exists a graph 
  $\graphstd_\cliquesize$ 
  with at most
  $
  \cliquesize \cliquevertexbound = 
  \bigoh{\cliquesize}
  $ 
  vertices 
  and a $3$-CNF formula
  $\cliqueformula[\cliquesize]{\graphstd_{\cliquesize}}$
  of size polynomial 
  in~$\cliquesize$    
    with the following properties:
  \begin{enumerate}
    \item 
      Resolution can refute 
      \cliqueformula{\graphstd_\cliquesize} in
      \ressize $2^{\bigoh{\cliquesize\log\cliquesize}}$
      and \reswidth $\cliquesize+1$.
    \item
      Any \SOS refutation of \cliqueformula{\graphstd_\cliquesize} requires
    \indexdegree $\cliquedegreebound \cliquesize$.
  \end{enumerate}
\end{theorem}

\begin{proof}
  Fix any positive $\epsilon < 1$ and let
  $\cliquevertexbound 
  = 
  2^{24\xorminvar}$, 
  $\cliquedegreebound 
  =
  \frac{\xordegreelb}{24}$ 
  and
  $\xorvarnum = \cliquesize\xorminvar$,
  where
  $\xorminvar$ and $\xordegreelb$ are the universal constants 
  from \refth{thm:3XOR}.
  To build the graph $\graphstd_\cliquesize$ we
  take a \xorformula formula $\xorstd$ on $\xorvarnum$ variables and
  $8\xorvarnum$ equations from the distribution in
  \refdef{def:randomxor}.
    Since $\xorvarnum \geq \xorminvar$, Theorem~\ref{thm:3XOR}
  implies that there is a formula in the support of the distribution
  that is unsatisfiable and that requires \sosdegree
  $\xordegreelb\xorvarnum$ to be refuted 
  in \SOS. 
    We fix~$\xorstd$~to be that formula and
  let
  $\graphstd_\cliquesize$ be the graph $\graphclausestd$ constructed
  as in \refdef{def:reduction}.
    Then
  $\graphclausestd$ is $\cliquesize$-partite,
  with 
  each part having at   most 
  $2^{{24\xorvarnum} / {\cliquesize}} = \cliquevertexbound$
  vertices, and
  the graph
  has no $\cliquesize$-clique because otherwise
  $\xorstd$ would be satisfiable.

  Suppose that there is an \SOS refutation of
  \Cliqueformula{\graphclausestd} of \indexdegree~$\sosdegreestd$.
  We want to argue that
  $\sosdegreestd\geq \cliquedegreebound \cliquesize$.
    Since $\graphclausestd$ is $\cliquesize$-partite, 
    by
  \refpr{stm:indexdegreeblock} 
  the
  formula
  \Blockformula{\graphclausestd} 
  also has 
  an \SOS refutation 
  in
  \indexdegree $\sosdegreestd$.
  By ~\reflem{lmm:reduction},
  this in turn yields an \SOS refutation of $\xorstd$ in \sosdegree~  ${24\sosdegreestd\xorvarnum} / {\cliquesize}$. 
  Now \refth{thm:3XOR} implies
  that
  ${24\sosdegreestd\xorvarnum} / {\cliquesize} \geq
  \xordegreelb\xorvarnum$,
  and hence
  $\sosdegreestd \geq 
  \frac{\xordegreelb}{24}\cliquesize
  = \cliquedegreebound\cliquesize$.

  To conclude the proof, we can just observe that 
  the resolution \reswidth and \ressize upper bounds are a direct
  application of \refpr{stm:cliqueupperbound}.
\end{proof}
 
\section{Size Lower Bounds from Relativization}
\label{sec:size-lower-bound-relativization}

Using the material developed in \refsec{sec:degree-lower-bound}, we
can now describe how to \introduceterm{relativize} formulas in order
to to amplify degree lower bounds to size lower bounds in \SOS\@. This
method works for formulas that are ``symmetric'' in a certain sense,
and so we start by explaining exactly what is meant by this.

\begin{definition}[Symmetric formula]
  \label{def:symmetric_formulas}
Consider a CNF formula $\formulastd$ on variables
$\varstd_{\domainidx,\myvec{\indexaux}}$,
where $\domainidx$ is an index
in some domain $\domainset$ and 
$\myvec{\indexaux}$ denotes a collection of other indices.
For every subset of indices
$
\signaturename = \signature{\indexdegreerun}
\subseteq \domainset
$  
we identify the
subformula $\formulastd_{\signaturename}$ of $\formulastd$ such that
each clause $\clc \in \formulastd_{\signaturename}$
mentions \emph{exactly} 
the indices in~$\signaturename$,
so that a formula~$\formulastd$ of \indexwidth~$\indexdegreestd$ can
be written as  
\begin{equation}
  \label{eq:symmetric_decomposition}
  \formulastd = 
  \Land^{\indexdegreestd}_{\indexdegreerun=0} 
  \Land_{\substack{\signaturename \subseteq \domainset \\
      \setsize{\signaturename} = \indexdegreerun}}
  \formulastd_{\signaturename}
  \eqperiod 
\end{equation}
We say that $\formulastd$ is 
\introduceterm{symmetric \wrt~$\domainset$} 
if it is
invariant with respect to permutations of $\domainset$,
\ie if
for every
 $\formulastd_{\signaturename} \subseteq \formulastd$ 
it also holds that
$
\formulastd_{\pi({\signaturename})}
\subseteq
\formulastd
$,
where $\pi$ is any permutation on~$\domainset$ and
$\pi\left({\signaturename}\right)$ is the 
set of images of the indices in~$\signaturename$.
Phrased differently, $\formulastd$~is symmetric \wrt~$\domainset$ if
for any permutation~$\pi$ on~$\domainset$ the \emph{syntactic} equality
$
\formulastd = 
\Land_{\signaturename \subseteq \domainset
}
\formulastd_{\pi(\signaturename)}
$
holds (where we recall that we treat CNF formulas as sets of clauses).
We apply this terminology for systems of polynomial equations and
inequalities in the same way.
\end{definition}

\ifthenelse{\boolean{conferenceversion}}{}{Let us illustrate \refdef{def:symmetric_formulas} by giving perhaps
the most canonical example of a formula that is symmetric in this
sense.

\begin{example}  \label{ex:php}
  Recall that
  the CNF encoding of the pigeonhole principle 
  with a set of pigeons $\domainset$
  and holes~$[n]$ claims
  that there is a mapping from pigeons in $\domainset$
  to holes such that no hole gets two pigeons. For every
  pigeon $\domainidx\in \domainset$ there is a clause 
  $\Lor_{\indexaux\in[n]} \varstd_{\domainidx,\indexaux}$ 
  and for every two distinct pigeons
  $\domainidx,\domainidx'$ 
  and hole
  $\indexaux$ there is a clause
    $\olnot{\varstd}_{\domainidx,\indexaux} \lor 
  \olnot{\varstd}_{\domainidx',\indexaux}$.  
  Since any permutation of the set of pigeons~$\domainset$ gives us
  back exactly the same set of clauses (only listed in a
  different order) the pigeonhole principle formula
  is symmetric \wrt~$\domainset$.
\end{example}

By now, the reader will already have guessed that another example of a
symmetric formula, which will be more interesting to us in the currect
context, is the $\cliquesize$-clique formula discussed in
\refsec{sec:degree-lower-bound}.
}

\begin{observation}
  The \cliqueformula{\graphstd} formula 
  in
  \refdef{def:cliqueformula}
  over variables
  $\cliquevar_{\domainidx,v}$ 
  is symmetric with respect to the
  indices $\domainidx \in [\cliquesize]$.
\end{observation}

Starting with any formula~$\formulastd$
symmetric \wrt a domain~$\domainset$, we can build a
family of similar formulas by varying the size of the domain.
If $\formulastd$ has \indexwidth~$\indexdegreestd$, 
then for each~$\indexdegreerun$, 
$0 \leq \indexdegreerun \leq \indexdegreestd$,
the subformulas
$\formulastd_{\signaturename}$ 
with $\setsize{\signaturename} = \indexdegreerun$
in~\eqref{eq:symmetric_decomposition}
are the same up to renaming of the domain indices in~$\signaturename$.
Hence, we can arbitrarily pick one such subformula to represent them
all, and denote it as $\formulastd_{\indexdegreerun}$.
The formulas~$\set{
  \formulastd_{\indexdegreerun}
}^{\indexdegreestd}_{\indexdegreerun=0}$
are completely determined by~$\formulastd$, and together
with~$\domainset$ they in turn completely determine~$\formulastd$. 
Using this observation, we can generalize the formula~$\formulastd$
over domain~$\domainset$ to any domain~$\domainset'$ with
$\setsize{\domainset'} \geq \indexdegreestd$
by defining 
$\formulaparam{\domainset'}$
to be the formula
\begin{equation}
  \label{eq:symmetric_generalization}
  \formulaparam{\domainset'}
    =
  \Land^{\indexdegreestd}_{\indexdegreerun=0} 
    \Land_{\substack{\signaturename \subseteq \domainset \\
      \setsize{\signaturename} = \indexdegreerun}}
  \formulastd_{\signaturename}
  \eqcomma
\end{equation}
where each $\formulastd_{\signaturename}$ for
$\setsize{\signaturename} = \indexdegreerun$ 
is an isomorphic
copy of $\formulastd_{\indexdegreerun}$ with its domain indices
renamed according to~$\signaturename$. 
Let us state some simple but useful facts that can be read off
directly from~\eqref{eq:symmetric_generalization}:
\begin{enumerate}
  \item 
    For any formula $\formulastd$ of \indexwidth~$\indexdegreestd$
    symmetric \wrt domain~$\domainset$, it holds that
    $\formulaparam{\domainset}$ is 
    (syntactically)
    equal to $\formulastd$.

  \item 
    For any domains $\domainset',\domainset''$ with
    $\setsize{\domainset'}=\setsize{\domainset''} \geq \indexdegreestd$, 
    the two formulas $\formulaparam{\domainset'}$ and
    $\formulaparam{\domainset''}$ are isomorphic.
    
  \item 
    For any $\domainset'' \supsetneq \domainset'$ with
    $\setsize{\domainset'} \geq \indexdegreestd$,
    the formula
    $\formulaparam{\domainset''}$ contains many isomorphic copies
    of~$\formulaparam{\domainset'}$. 
\end{enumerate}

When we want to emphasize the domain~$\domainset$ of a
formula~$\formulastd$ in what follows, we will denote the
formula~$\formulastd$ as~$\formulaparam{\domainset}$. 
When the domain is 
$\domainset=[\domainsize]$, 
we abuse notation slightly and write
$\formulaparam{\domainsize}$ instead
of~$\formulaparam{[\domainsize]}$.
As discussed above, from a symmetric formula~$\formulastd$ of
\indexwidth~$\indexdegreestd$
we can obtain a well-defined sequence of formulas
$\formulaparam{\domainsize}$ for all $\domainsize \geq \indexdegreestd$.
We say that the \introduceterm{unsatisfiability threshold}
of such a sequence of formulas
is the least
$\domainsize$ such that 
$\formulaparam{\domainsize}$ is unsatisfiable.
\ifthenelse{\boolean{conferenceversion}}           
{}
{For instance, 
  the pigeonhole principle formula in \refex{ex:php} has
  unsatisfiability threshold $n+1$.}

\subsection{\Relativization of Symmetric Formulas}

Given a formula
$\formulastd = \formulastd[\domainlarge]$
symmetric \wrt $[\domainlarge]$
and a parameter~$\domainsmall < \domainlarge$,
we now want to define
the \introduceterm{$\domainsmall$-\relativization}
of~$\formulastd[\domainlarge]$,
which is intended to encode the claim that that there
exists a subset $\domainset \subseteq [\domainlarge]$
of size
$\setsize{\domainset} \geq \domainsmall$
such that the subformula
$\formulaparam{\domainset} \subseteq \formulaparam{\domainlarge}$ is
satisfiable. 
We remark that a CNF formula encoding such a claim will be
unsatisfiable when $\domainsmall$ is at least the 
unsatisfiability threshold of~$\formulastd$.

In order to express the existence of the subset $\domainset$ we use
\introduceterm{\guardvariable{}s}
$\relguard_{1}, \relguard_{2}, \ldots, \relguard_{\domainlarge}$ 
as indicators of membership in the subset and encode the constraint
on the subset size
$
\setsize{\domainset} =
\sum^{\domainlarge}_{\domainidx=1}\relguard_{\domainidx} \geq
\domainsmall
$
as described in the next definition.

\begin{definition}
  \label{def:threshold_formula}
  The \introduceterm{threshold-$\domainsmall$ formula} for variables
  $
  \myvec{\relguard}= 
  \set{\relguard_{1},\ldots,\relguard_{\domainlarge}}
  $ 
  is the $3$-CNF formula 
  $\THR{\domainsmall}{\myvec{\relguard}}$ 
  that consists of the clauses
  \begin{subequations}\label{eq:thrformula}
  \begin{align}
      &
      \thrvaraux_{\thridxaux,0}
    && 
       \thridxaux\in [\domainsmall],
       \label{eq:thrformula_first}
    \\
    &
      \olnot{\thrvaraux}_{\thridxaux,\domainidx-1} \lor
      \thrvarstd_{\thridxaux,\domainidx} \lor
      \thrvaraux_{\thridxaux,\domainidx}  
    &&
       \thridxaux\in[\domainsmall], \domainidx\in [\domainlarge],
       \label{eq:thrformula_second}
    \\
    &
      \olnot{\thrvaraux}_{\thridxaux,\domainlarge} 
    && 
       \domainidx\in [\domainlarge],
       \label{eq:thrformula_third}
    \\
    &
      \olnot{\thrvarstd}_{\thridxaux,\domainidx} \lor
      \olnot{\thrvarstd}_{\thridxaux',\domainidx}
    &&
      \thridxaux,\thridxaux'\in [\domainsmall],
      \thridxaux \neq \thridxaux',
      \domainidx\in[\domainlarge],
      \label{eq:thrformula_injectivity}
    \\
    &
      \olnot{\thrvarstd}_{\thridxaux,\domainidx} \lor
      \relguard_{\domainidx}
    &&
       \thridxaux\in [\domainsmall],
       \domainidx\in[\domainlarge]\eqperiod
       \label{eq:thrformula_count}
  \end{align}
  \end{subequations}
\end{definition}

To see that
$\THR{\domainsmall}{\myvec{\relguard}}$ 
indeed enforces a cardinality constraint, note that the
variables 
$\thrvarstd_{\thridxaux,\domainidx}$ 
encode a mapping
between $[\domainsmall]$ and $[\domainlarge]$
(with
$\thrvarstd_{\thridxaux,\domainidx}$ 
being true if and only if
$\thridxaux$ maps to~$\domainidx$).
\ifthenelse{\boolean{conferenceversion}}{}{The clauses~\eqref{eq:thrformula_first}--\eqref{eq:thrformula_third}
force every $\thridxaux\in[\domainsmall]$ to have an image in
$[\domainlarge]$, since they form the $3$-CNF representation of clauses
$\Lor_{\domainidx} \thrvarstd_{\thridxaux,\domainidx}$.
The clauses~\eqref{eq:thrformula_injectivity} forbid two distinct
elements of $[\domainsmall]$ to have the same image,
so there must be at least $\domainsmall$ elements in the range of the
map, and for each of them the corresponding \guardvariable{} must be
true because of the clauses~\eqref{eq:thrformula_count}.
}We will need the following properties of the threshold formula.

\begin{observation}
  \label{obs:thrupperbound}
  The formula $\THR{\domainsmall}{\myvec{\relguard}}$ 
  in
  \refdef{def:threshold_formula}
  has the following properties:
    \begin{enumerate}
    \item 
      $\THR{\domainsmall}{\myvec{\relguard}}$ 
      has size polynomial 
      in both $\domainsmall$ and $\domainlarge$.
      
    \item
      For any partial assignment to $\myvec{\relguard}$ with at
      least $\domainsmall$ ones there is an assignment to the extension
      variables that satisfies $\THR{\domainsmall}{\myvec{\relguard}}$.
    \item 
      There is a resolution refutation of the set of clauses
      $
      \THR{\domainsmall}{\myvec{\relguard}}
      \union
      \Setdescr
      {\Lor_{\domainidx\in \domainset} \olnot{\relguard}_{\domainidx}}
      {\domainset \subseteq [\domainlarge],\, \setsize{\domainset}=k}
      $
      of \ressize
      $\Bigoh{k\domainlarge^{\domainsmall}}$ and width $\domainsmall+1$.
  \end{enumerate}
\end{observation}
\begin{proof}
  The first two items are immediate.  
    In order to show the third item we can first derive each clause
  $\olnot{\thrvarstd}_{1,\domainidx_{1}} \lor \ldots \lor
  \olnot{\thrvarstd}_{\domainsmall,\domainidx_{\domainsmall}}$ 
  by resolving   
  $\olnot{\relguard}_{\domainidx_{1}} \lor \ldots \lor
  \olnot{\relguard}_{\domainidx_{\domainsmall}}$ 
  with clauses of the form~\eqref{eq:thrformula_count}, and then apply
  \reflem{lmm:bruteforce}.
\end{proof}

Using the formula in \refdef{def:threshold_formula} to encode
cardinality constraints on subsets, we can now define formally what we
mean by the relativization of a symmetric formula.

\begin{definition}  [Relativization]
  \label{def:relativization}
  Given a CNF formula $\formulastd$ symmetric \wrt a
  domain~$[\domainlarge]$ and a parameter
  $\domainsmall < \domainlarge$, 
  the
  \introduceterm{$\domainsmall$-relativization} 
  (or \introduceterm{$\domainsmall$-relativized formula})
  $\relstd$ is the
  formula consisting of
  \begin{enumerate}
    \item 
      \label{item:rel_threshold}
      the threshold formula $\THR{\domainsmall}{\myvec{\relguard}}$
      over \guardvariable{}s
            $\myvec{\relguard}= \set{\relguard_{1},\ldots,\relguard_{\domainlarge}}$;
    \item 
      \label{item:rel_guard}
      a \introduceterm{\guardedclause{}} 
            $\olnot{\relguard}_{\domainidx_{1}} 
      \lor
      \ldots
      \lor
      \olnot{\relguard}_{\domainidx_{\indexdegreerun}}
      \lor
      \clc$
            for each clause $\clc\in \formulaparam{\domainlarge}$, 
      where $\signature{\indexdegreerun}$ are the indices mentioned by $\clc$. 
  \end{enumerate}
\end{definition}

Since we are dealing with refutations of unsatisfiable formulas, it
will always be the case that the parameter~$\domainsmall$ in
\refdef{def:relativization} is at least the unsatisfiability threshold
of~$\formulastd$.
An important property of relativized formulas is that the hardness of
$\relstd$ scales nicely with~$\domainlarge$. 
In particular, if 
$\formulaparam{\domainsmall}$ is not too hard, then the
relativization~$\relstd$ also is not too hard.

\begin{proposition}
  \label{stm:relupperbound}
  If $\formulaparam{\domainsmall}$ has a resolution refutation of 
  \ressize~$\ressizestd$ and \reswidth~$\reswidthstd$, then $\relstd$
  has a resolution refutation of \ressize 
  $\ressizestd\cdot
  \binom{\domainlarge}{\domainsmall}+\Bigoh{k\domainlarge^\domainsmall}$ 
  and \reswidth $\reswidthstd+\domainsmall$. 
\end{proposition}

\begin{proof}
  For every set $\domainset \subseteq [\domainlarge]$ with
  $\setsize{\domainset}=\domainsmall$
  we
  show how to derive
  \begin{equation}
   \label{eq:large_domain}
   \Lor_{\domainidx\in\domainset }\olnot{\relguard}_{\domainidx}
  \end{equation}
  in \ressize $\ressizestd+1$ and \reswidth
  $\reswidthstd+\domainsmall$ from $\relstd$.
    Without loss of generality 
  (because of symmetry)
  we assume that $\domainset=[k]$, so that
  we want to derive  
  $\olnot{\relguard}_{1} \lor 
  \cdots
  \lor \olnot{\relguard}_{\domainsmall}$.
    Consider the assignment
  $\partassign=\{\relguard_{1}=1,\ldots,\relguard_{\domainsmall}=1\}$.
    In the restricted formula $\restrict{\relstd}{\partassign}$ the
  \guardedclause{}s in \refdef{def:relativization},
  item~\ref{item:rel_guard},
  with all indices in $[\domainsmall]$
  become the clauses of $\formulaparam{\domainsmall}$, which has a
  refutation of \ressize $\ressizestd$ and \reswidth $\reswidthstd$.
    Thus the clause
  $\olnot{\relguard}_{1} \lor
  \cdots
  \lor
  \olnot{\relguard}_{\domainsmall}$
  can be derived in \ressize~$\ressizestd+1$ and
  \reswidth~$\reswidthstd+\domainsmall$
  from $\relstd$ 
  by
  Fact~\ref{fact:proofweakening}.
      After we have derived all clauses
  of the form~  \eqref{eq:large_domain} in this way, 
  we can obtain
  the empty clause in \reswidth $\domainsmall+1$ and in \ressize at
  most $\Bigoh{k\domainlarge^\domainsmall}$ using
  Observation~\ref{obs:thrupperbound}.
  \end{proof}

\subsection{Random Restrictions and Size Lower Bounds}
\label{sec:size-lower-bounds}

To prove size lower bounds on refutations of relativized
formulas~$\relstd$ we use random restrictions sampled as follows. 

\begin{definition}  [Random restrictions for relativized formulas]
  \label{def:restriction}
  Given a relativized formula~$\relstd$, we define a distribution
  $\padistribution$ of partial assignments over the variables of this
  formula by the following process.
  \begin{enumerate}
    \item
      Pick uniformly at random a set
      $\domainset\subseteq[\domainlarge]$ of size $\domainsmall$.

    \item 
      Fix $\relguard_{\domainidx}$ to $1$ if
      $\domainidx\in\domainset$ and to $0$ otherwise.
    
    \item
      Extend this to any assignment to the remaining variables of the
      formula~$\THR{\domainsmall}{\myvec{\relguard}}$ that
      satisfies~      this threshold formula.
      
    \item
      For every variable
      $\varstd_{\indexstd,\myvec{\indexaux}}$
            that has index
      $\domainidx\not\in \domainset$, 
      fix
      $\varstd_{\indexstd,\myvec{\indexaux}}$
      to $0$ or~$1$ uniformly and independently at random.
    \item
      All remaining variables
      $\varstd_{\indexstd,\myvec{\indexaux}}$
            for the indices
      $\domainidx\in \domainset$
      are left unset.
  \end{enumerate}
\end{definition}

It is straightforward to verify that 
the distribution
$\padistribution$ is constructed
in such a way as to give us back 
$\formulaparam{\domainsmall}$ 
from~$\relstd$.

\begin{observation}
  \label{obs:restriction-yields-Fk}
  For any relativized formula~$\relstd$ and any
  $\partassign\in\padistribution$ it holds that
  $\restrict{\relstd}{\partassign}$ is equal to
  $\formulaparam{\domainsmall}$ 
  up to renaming of variables.
\end{observation}

The key technical ingredient in the size lower bound on \sumsofsquares
proofs is
the following property of the distribution $\padistribution$,
which was proven in~\cite{AMO13LowerBounds,ALN14NarrowProofsECCC}
but is rephrased below using the notation and terminology 
in this paper.
\ifthenelse{\boolean{conferenceversion}}           
{}
{We also provide a brief proof sketch just to give the reader a sense
  of how the argument goes.}

\begin{lemma}[    \ifthenelse{\boolean{conferenceversion}}    {\cite{ALN14NarrowProofsECCC,AMO13LowerBounds}}    {\cite{AMO13LowerBounds,ALN14NarrowProofsECCC}}  ]
  \label{lmm:StrongRestriction}
  Let $\domainsmall,\genericint,\domainlarge$ be
  positive integers
  such that 
  $ \domainlarge\geq 16$ and 
  $\genericint \leq \domainsmall \leq \domainlarge/(4\log \domainlarge)$.
    Let $\monomialstd$ be a monomial over the variables of $\relstd$
  and let $\partassign$ be a random restriction
  sampled from
  the distribution $\padistribution$
  in
  \refdef{def:restriction}.
    Then the \indexdegree of 
  $\restrict{\monomialstd}{\partassign}$ 
  is less than
  $\genericint$ 
  with probability at least
  $1 - \RestrictionBoundInline{\genericint}$.
\end{lemma}

\ifthenelse{\boolean{conferenceversion}}{}{\begin{proof}[Proof sketch]
  Ley $\genericint'$ be the \indexdegree of $\monomialstd$. 
  The restriction $\partassign$ will set independently and
  uniformly at random at least $\genericint'-\domainsmall$ of its
  variables, so if
    $(\genericint'-\domainsmall)$ is larger than
  ${\genericint \log \domainlarge}$, the restricted monomial
  $\restrict{\monomialstd}{\partassign}$ is non zero with probability
  at most $1/\domainlarge^{\genericint}$.
    Otherwise we upper bound the probability that
  $\restrict{\monomialstd}{\partassign}$ has \indexdegree $\genericint$
  with the probability that the $\genericint'$
  indices in $\monomialstd$ contain $\genericint$ of the
  $\domainsmall$ surviving indices.
    By a union bound
  this probability is at \mbox{most $\RestrictionBoundInline{\genericint}$}.
\end{proof}
}
Using
\reflem{lmm:StrongRestriction}, it is now straightforward to show that
relativization amplifies degree lower bounds to size lower bounds.

\begin{theorem}
  \label{thm:size-lower-bounds}
  Let $\domainsmall,\genericint,\domainlarge$ be 
  positive integers
  such that 
  $ \domainlarge\geq 16$ and 
  $\genericint \leq \domainsmall \leq \domainlarge/(4\log \domainlarge)$.
    If 
  the CNF formula~  $\formulaparam{\domainsmall}$ requires \sumsofsquares
  refutations of \indexdegree $\genericint$,
    then the relativized formula~  $\relstd$ requires \sumsofsquares refutations of \sossize
  $\InverseRestrictionBoundInline{\genericint}$.
\end{theorem}

\begin{proof}
            Suppose that there is a \sumsofsquares refutation of $\relstd$
  in size~$\sossizestd$, \ie containing $\sossizestd$~monomials. 
              For $\partassign$ sampled from~$\padistribution$, 
  we see that the probability
  that some monomial in the refutation restricted by~$\partassign$
  has \indexdegree 
  at least
  $\genericint$ is at most
  \begin{equation}
    \label{eq:unionbound_lower_bound}
    \sossizestd \cdot \RestrictionBoundDisplay{\genericint}
  \end{equation}
  by appealing to \reflem{lmm:StrongRestriction} and taking a union bound.

  As noted in \refobs{obs:restriction-yields-Fk},
  the formula~$\restrict{\relstd}{\partassign}$ is equal to
  $\formulaparam{\domainsmall}$ 
  up to renaming of variables, and so it cannot
  have a refutation of \indexdegree~$\genericint$ or less.
    This implies that the bound on the
  probability~\eqref{eq:unionbound_lower_bound} is greater than one,
  and thus we obtain
  \begin{equation}
    \sossizestd > \InverseRestrictionBoundDisplay{\genericint}
    \eqcomma
  \end{equation}
  which proves the theorem.
\end{proof}

\ifthenelse{\boolean{conferenceversion}}{}{\subsection{Statement of Main Result and Discussion of Possible Improvements}
\label{sec:mainresult}
}

Putting everything together, we can establish the formal version of our
main results in
\refth{th:main-result-informal}
as follows.

\begin{theorem}
  \label{thm:main-results}
  Let $\domainsmall = \domainsmall(\domainlarge)$ be any monotone non-decreasing
  integer-valued function such that 
  $\domainsmall(\domainlarge) \leq \domainlarge/(4 \log \domainlarge)$. 
    Then there is a family of $4$-CNF formulas
  $\{\targetformula\}_{\domainlarge\geq 1}$
  with   $\Bigoh{\domainsmall\domainlarge^{2}}$ clauses over
  $\bigoh{\domainsmall\domainlarge}$ variables 
  such that:
  \begin{enumerate}
    \item
      Resolution can refute
      $\targetformula$
      in \ressize
      $\domainsmall^{\bigoh{\domainsmall}}m^{\domainsmall}$ and
      \reswidth $2\domainsmall+1$.
    \item 
      Any \sumsofsquares refutation of $\targetformula$ requires
      \sossize 
      $\Bigomega{\domainlarge^{\cliquedegreebound\domainsmall} /
        (4\domainsmall\log \domainlarge)^{\domainsmall}}$,
      where $\cliquedegreebound$ is a universal constant.
  \end{enumerate}
\end{theorem}

\begin{proof}
  Let $\graphstd$ be a graph with properties as in
  \refth{thm:degreelowerbound} 
  and let $\formulaparam{\domainsmall}$ be the CNF formula
  \cliqueformula{\graphstd} 
  in \refdef{def:cliqueformula}.
  Since $\formulaparam{\domainsmall}$ is symmetric,
  we can relativize it as in
  \refdef{def:relativization}
  to obtain
  $\relstd$, which will be our $4$-CNF formula~$\targetformula$.
  \Refth{thm:degreelowerbound} says that
  $\formulaparam{\domainsmall}$ has a resolution refutation of
  \ressize~$\domainsmall^{\bigoh{\domainsmall}}$ and
  \reswidth~$\domainsmall+1$, and appealing to
  \refpr{stm:relupperbound} we get 
  a resolution refutation of~$\targetformula$ in size 
  $\domainsmall^{\bigoh{\domainsmall}}m^{\domainsmall}$ and \reswidth
  $2\domainsmall+1$. 
  Since we have a \indexdegree lower bound
  of $\cliquedegreebound\domainsmall$
  for refuting
  $\formulaparam{\domainsmall}$ 
  according to
  \refthm{thm:degreelowerbound},
  we can use
  \refth{thm:size-lower-bounds} to deduce that the required  \sossize
  to refute $\targetformula$ in \sumsofsquares  is at least
  $\Bigomega{
    \domainlarge^{\cliquedegreebound\domainsmall}  / 
    (4\domainsmall \log \domainlarge)^{\domainsmall}
  }$.
  The theorem follows.
\end{proof}

We remark that 
straightforward calculations show that when
$\domainsmall(\domainlarge) = 
\Bigoh{\domainlarge^{\lowerboundlimit}}$ for
$\lowerboundlimit<\cliquedegreebound$
the upper bound 
in   \refthm{thm:main-results} is 
$\domainlarge^{\bigoh{\domainsmall}}$ 
and the lower bound is
$\domainlarge^{\bigomega{\domainsmall}}$.

\ifthenelse{\boolean{conferenceversion}}{}{Let us now discuss a couple of the parameters in
\refth{thm:main-results} 
and how they could be improved slightly.
We stated our main theorem for $4$-CNF formulas, since that is the
clause size that results naturally from our construction. However, if
one wants to minimize the clause width and obtain an analogous result
for \mbox{$3$-CNF} formulas this is also possible to achieve, just as
was done in~\cite{ALN14NarrowProofsECCC} for other proof systems.
To prove a version of \refth{thm:main-results} for $3$-CNF formulas
we need a simple but rather ad-hoc variation of the relativization
argument presented above. 
Let us briefly describe what modifications are needed.

The way we presented the construction above, we started with the
$3$-CNF formula
\cliqueformula{\graphstd} and then applied relativization, which
turned the 
clauses~\eqref{eq:index_clause_cnfA}--\eqref{eq:index_clause_cnfC}
into the $4$-CNF formula
\begin{subequations}
  \label{eq:index_clause_relativized}
  \begin{align}
    \label{eq:index-clause-relativized-a}
    &
      \olnot{\relguard}_{\domainidx} \lor \cliqueaux_{\domainidx,0}
    &&
       \text{$\domainidx\in[\cliquesize]$,}
    \\
    \label{eq:index-clause-relativized-b}
    &
      \olnot{\relguard}_{\domainidx} \lor 
      \olnot{\cliqueaux}_{\domainidx,{(\indexaux-1)}} \lor
      \cliquevar_{\domainidx,v_\indexaux} \lor
      \cliqueaux_{\domainidx,\indexaux}
    &&
       \text{$\domainidx\in[\cliquesize]$, $\indexaux\in[\cliquevertexnum]$,}
    \\
    \label{eq:index-clause-relativized-c}
    &
      \olnot{\relguard}_{\domainidx} \lor \olnot{\cliqueaux}_{\domainidx,\cliquevertexnum}
    &&
       \text{$\domainidx\in[\cliquesize]$.}
  \end{align}
\end{subequations}
An alternative approach 
would be to first encode 
\cliqueformula{\graphstd} with
wide clauses
$\Lor^{\cliquevertexnum}_{\indexaux=1}
\cliquevar_{\domainidx,v_{\indexaux}}$ instead of 
clauses
of the form~\eqref{eq:index_clause_cnfA}--\eqref{eq:index_clause_cnfC},
relativize this new, wide formula, and then convert the
relativized formula into $3$-CNF using extension variables.
Instead of
clauses~\refeq{eq:index-clause-relativized-c}--\refeq{eq:index-clause-relativized-c},
this would yield the collection of clauses
\begin{subequations}
  \label{eq:wide_clause_relativized}
  \begin{align}
  \label{eq:wide-clause-relativized-a}
    &
      \olnot{\relguard}_{\domainidx} \lor \cliqueaux_{\domainidx,0}
    &&
       \text{$\domainidx\in[\cliquesize]$,}
    \\
  \label{eq:wide-clause-relativized-b}
    &
      \olnot{\cliqueaux}_{\domainidx,{(\indexaux-1)}} \lor
      \cliquevar_{\domainidx,v_\indexaux} \lor
      \cliqueaux_{\domainidx,\indexaux}
    &&
       \text{$\domainidx\in[\cliquesize]$, $\indexaux\in[\cliquevertexnum]$,}
    \\
  \label{eq:wide-clause-relativized-c}
    &
      \olnot{\cliqueaux}_{\domainidx,\cliquevertexnum}
    &&
       \text{$\domainidx\in[\cliquesize]$.}
  \end{align}
\end{subequations}

This causes a small technical problem in that 
some of these clauses mention
$\domainidx\in[\domainlarge]$ but lack the literal
$\olnot{\relguard}_{\domainidx}$, and so a random restriction sampled
as in \refdef{def:restriction} may actually falsify these clauses.
The solution to this is to change the random assignment so that 
when $\relguard_{\domainidx}=0$,
we fix each $\cliquevar_{\domainidx,v_\indexaux}$ uniformly at random
in~$\set{0,1}$, set each $\cliqueaux_{\domainidx,{(\indexaux-1)}} $
equal to the value assigned to~$\cliquevar_{\domainidx,v_\indexaux}$,
and finally fix $\cliqueaux_{\domainidx,\cliquevertexnum}$ to~$0$.
The new restriction satisfies all
clauses~\refeq{eq:wide-clause-relativized-a}--\refeq{eq:wide-clause-relativized-c},
and the proof of
\reflem{lmm:StrongRestriction}
still goes through.

Another parameter in \refth{thm:main-results} that could be improved
is the value of~$\cliquedegreebound$, which determines how tightly the
size lower bound matches the upper bound implied by width/degree and
also 
how 
high we can push $\cliquesize(\domainlarge)$.
In our reduction from a \xorformula  formula~$\xorstd$ to the clique
formula \Cliqueformula{\graphclausestd} we start by splitting the
$8\xorvarnum$~constraints into $\cliquesize$ blocks.
The vertices in each block correspond to assignments to
$24\xorvarnum/\cliquesize$~variables, and because of this an \SOS
refutation in \indexdegree~$\sosdegreestd$ of
\Cliqueformula{\graphclausestd} can be converted to a refutation in
degree~$24\sosdegreestd\xorvarnum/\cliquesize$ of~$\xorstd$.

If we want to obtain a more efficient reduction, we could instead split
the \emph{$\xorvarnum$~variables}, rather than the
$8\xorvarnum$~constraints,
into $\cliquesize$~parts.
In this way each vertex in $\graphclausestd$ would correspond to an
assigment to $\xorvarnum/\cliquesize$~variables, and an \SOS
refutation in \indexdegree $\sosdegreestd$ would translate to a
refutation of~$\xorstd$ in degree~$\sosdegreestd\xorvarnum/\cliquesize$.
But now we cannot reduce to the clique problem anymore.
Splitting with respect to constraints allows us to enforce pairwise
consistency between vertices in different blocks referring to common
variables. 
When splitting with respect to variables, the vertices in different
blocks correspond to partial assigments on disjoint domains and so are
always pairwise compatible. However, we must still require that these
partial assignments are consistent with the constraints in~$\xorstd$.
Each such constraint refers to up to three blocks. Thus, any satisfying
assignment to~$\xorstd$ corresponds to $\cliquesize$~vertices such
that no triple of vertices violates an \xorformula constraint.
This reduces to the problem of finding a
\mbox{$\cliquesize$-hyperclique} in a \mbox{$3$-uniform} hypergraph.
The rest of the reduction can be made to work as in
\reflem{lmm:reduction}.
In the end we get an analogous result of that in
\refth{thm:degreelowerbound} but with 
$\cliquedegreebound$ equal to $\xordegreelb$ instead
of~$\frac{\xordegreelb}{24}$, which also improves 
\refth{thm:main-results}.
In this paper we instead presented a reduction to the $\cliquesize$-clique
problem for standard graphs, partly because we believe that a
\sosdegree lower bound for this problem can be considered to be of
independent interest. 
} 
\section{Concluding Remarks}
\label{sec:conclusion}

In this paper, we show that using Lasserre semidefinite programming
relaxations to find degree-$\sosdegreestd$ \sumsofsquares proofs is
optimal up to  
constant factors in the exponent of the running time. More precisely,
we show that there are constant-width CNF formulas on~$\varnum$ variables
that are refutable in sums-of-squares in degree~$\sosdegreestd$ but require
proofs of size $\varnum^{\bigomega{\sosdegreestd}}$.

As for so many other results for the \sumsofsquares proof system, in
the end our proof boils down to a reduction from \xorformula using
Schoenebeck's version~\cite{Schoenebeck08LinearLevel}
of Grigoriev's degree lower bound~\cite{Grigoriev01LinearLowerBound}.
It would be very interesting to obtain other SOS degree lower bounds
by different means than by reducing from Grigoriev's results for
\xorformula and knapsack.  

Another interesting problem would be to prove 
average-case
SOS degree lower bound
for \mbox{$\cliquesize$-clique} formulas over Erd\H{o}s--Rényi random graphs,
or size lower bounds for (non-relativized) \mbox{$\cliquesize$-clique}
formulas over any graphs.  In this context, it might be worth to
point out that the problem of establishing proof size lower bounds
for \mbox{$\cliquesize$-clique} formulas for constant~$\cliquesize$,
which has been discussed, for instance, in~\cite{BGLR12Parameterized},
still remains open even for the resolution proof system
(although lower bounds have been shown for
tree-like resolution in~\cite{BGL13ParameterizedDPLL}
and for full resolution 
for a version of clique formulas
using a different encoding more amenable
to lower bound techniques in~\cite{LPRT13ComplexityRamsey}).
 
\section*{Acknowledgements}

We are grateful to Albert Atserias for numerous discussions about (and
explanations of) Lasserre/SOS and other LP and SDP hierarchies, as
well as for help with correcting some references in a preliminary
version of this manuscript.  We thank Per Austrin for valuable
suggestions and feedback during the initial stages of this work, and
Michael Forbes for comments on an early version of the overall proof
construction.  Finally, we are thankful for the comments from the
anonymous reviewers, which helped improve the exposition in this paper
considerably.

The authors were funded by the
European Research Council under the European Union's Seventh Framework
Programme \mbox{(FP7/2007--2013) /} ERC grant agreement no.~279611.
\TheauthorJN 
was also supported by
Swedish Research Council grants 
\mbox{621-2010-4797}
and
\mbox{621-2012-5645}.

\bibliography{ArXiv-SizeSumsOfSquares}

\bibliographystyle{alpha}

\end{document}